\documentclass[a4paper,english]{article}
\usepackage[margin=2.5cm]{geometry}
\usepackage[utf8]{inputenc}
\usepackage{subcaption}

\usepackage{stmaryrd} 

\SetSymbolFont{stmry}{bold}{U}{stmry}{m}{n}
\usepackage[newdocumentcommand,textwidth=1.5cm,
            paralist=false,todo=hide]{generic}
\usepackage{thm-restate} 
\usepackage{chngcntr} 
\usepackage[electronic,xcolor,hyperref,notion,quotation]{knowledge} 
\usetikzlibrary{decorations.pathreplacing,tikzmark,bending,patterns,calc} 

\newrobustcmd{\angelo}[2][]{{\color{black}\todo[color=cyan!30,#1]{{\bf Angelo:} #2}}}
\newrobustcmd{\dario}[2][]{{\color{black}\todo[color=blue!30,#1]{{\bf Dario:} #2}}}
\newrobustcmd{\gabriele}[2][]{{\color{black}\todo[color=magenta!25,#1]{{\bf Gabriele:} #2}}}
\newrobustcmd{\pietro}[2][]{{\color{black}\todo[color=yellow!20,#1]{{\bf Pietro:} #2}}}

\renewcommand\paragraph[1]{\medskip\noindent\emph{#1.}~}

\newsavebox{\wraptext}
\newsavebox{\wrapfig}
\newsavebox{\deftext}
\savebox{\deftext}{Xp}

\knowledgestyle*{intro notion}{color=blue!50!black, emphasize}
\knowledgestyle*{notion}{color=black}
\knowledgestyle*{unknown (cont)}{color=black}

\makeatletter
\patchcmd\@combinedblfloats{\box\@outputbox}{\unvbox\@outputbox}{}{%
}%
\makeatother

\tikzset{every node/.style={font=\vphantom{Ag}, minimum size=1mm, inner sep=0pt, outer sep=0pt}}

\tikzstyle{reverseclip}=[insert path={(current page.north east) --
  (current page.south east) --
  (current page.south west) --
  (current page.north west) --
  (current page.north east)}
]

\tikzstyle{short}=[shorten >=8pt, shorten <=8pt]
\tikzstyle{dot}=[draw, shape=circle, minimum size=1.5mm, fill, 
                 style={font=\vphantom{}}]
\tikzstyle{bluebox}=[shape=rectangle, fill=cyan!50, minimum size=4mm, inner sep=2pt, outer sep=0pt]
\tikzstyle{ubrace} = [draw, thick, decoration={brace, mirror, raise=0.0cm}, decorate, 
                      every node/.style={anchor=north, yshift=-0.1cm}]

\tikzstyle{outline} = [preaction={-, draw, shorten >=2pt, shorten <=2pt, line width=2pt, white}]

\tikzstyle{arrow} = [shorten >=5pt, shorten <=0pt, arrows = {->[bend,scale=1.5]}]
\tikzstyle{directededge} = [shorten >=5pt, shorten <=0pt, arrows = {-latex'}]
\tikzstyle{over} = [preaction={-, draw, line width=8pt, white}]
\tikzstyle{thinover} = [preaction={-, draw, line width=4pt, white}]
\tikzstyle{partialarrow} = [shorten >=5pt, shorten <=0pt, arrows = {->[bend,round,scale=1.5,harpoon]}]
\tikzstyle{surjectivearrow} = [shorten >=5pt, shorten <=0pt, arrows = {-To[bend,round,scale=1.5]To[bend,round,scale=1.5]}]
\tikzstyle{injectivearrow} = [shorten >=5pt, shorten <=0pt, arrows = {>[bend,round,scale=1.5]->[bend,round,scale=1.5]}]
\tikzstyle{mapsto} = [shorten >=5pt, shorten <=0pt, |->]

\tikzstyle{loop above} = [shorten >=3pt, shorten <=3pt, loop, looseness=10, in=60, out=120]
\tikzstyle{loop below} = [shorten >=3pt, shorten <=3pt, loop, looseness=10, in=-60, out=-120]

\pgfdeclaredecoration{round decoration}{initial}{
  \state{initial}[width=0pt, next state=bracket]{
    \pgfpathmoveto{\pgfpoint{0pt}{1.5pt}}
  }
  \state{bracket}[width=\pgfdecoratedremainingdistance, next state=final]{
    \pgfpathlineto{\pgfpoint{0pt}{1.5pt}}
    \pgfpathcurveto{\pgfpoint{2pt}{1.5pt}}{\pgfpoint{2pt}{-1.5pt}}{\pgfpoint{0pt}{-1.5pt}}
    \pgfpathlineto{\pgfpoint{0pt}{-1.5pt}}
  }
  \state{final}{
    \pgfpathlineto{\pgfpointdecoratedpathlast}
  }
}

\tikzstyle{brace} = [shorten >=0pt, shorten <=0pt, decorate, 
                     decoration={brace, amplitude=10pt}]
\tikzstyle{above round brace} = [inherit options/.code=
                                     {\csname tikz@options\endcsname},
                                 inherit options,
                                 decorate,
                                 decoration={show path construction,
                                             lineto code={
                                                \draw [line cap=round] 
                                                (\tikzinputsegmentfirst) 
                                                arc (180:90:0.15cm) --
                                                ($(\tikzinputsegmentlast) + 
                                                  (-0.15cm,0.15cm)$)
                                                arc (90:0:0.15cm);
                                            }
                                        }]
\tikzstyle{below round brace} = [inherit options/.code=
                                     {\csname tikz@options\endcsname},
                                 inherit options,
                                 decorate,
                                 decoration={show path construction,
                                             lineto code={
                                                \draw [line cap=round] 
                                                (\tikzinputsegmentfirst) 
                                                arc (-180:-90:0.15cm) --
                                                ($(\tikzinputsegmentlast) + 
                                                  (-0.15cm,-0.15cm)$)
                                                arc (-90:0:0.15cm);
                                            }
                                        }]

\definecolor{nicecyan}{HTML}{006165}
\definecolor{nicered}{HTML}{DB3A34}
\definecolor{nicegreen}{HTML}{6D972E}

\newcommand\timedomain{"N@time domain"}
\newcommand\timedomainbis{"N'@time domain"}
\WithSuffix\newcommand\timedomain*{N}
\WithSuffix\newcommand\timedomainbis*{N'}

\newcommand\intervaldomain{"\bbI(\timedomain*)@interval domain"}
\newcommand\intervaldomainbis{"\bbI(\timedomainbis*)@interval domain"}
\WithSuffix\newcommand\intervaldomain*{\bbI(\timedomain*)}
\WithSuffix\newcommand\intervaldomainbis*{\bbI(\timedomainbis*)}

\newcommand\intervalstructure{"(\intervaldomain*,\sigma)@interval structure"}
\newcommand\intervalstructurebis{"(\intervaldomainbis*,\sigma')@interval structure"}
\WithSuffix\newcommand\intervalstructure*{(\intervaldomain*,\sigma)}
\WithSuffix\newcommand\intervalstructurebis*{(\intervaldomainbis*,\sigma')}

\newcommand\prefix{\mathbin{"<_B@prefix relation"}}
\newcommand\suffix{\mathbin{"<_E@suffix relation"}}

\newcommand\signature{"\Sigma@signature"}
\newcommand\signaturebis{"\Sigma'@signature"}
\WithSuffix\newcommand\signature*{\Sigma}
\WithSuffix\newcommand\signaturebis*{\Sigma'}

\newcommand{\BE}{"\mathsf{BE}@\BE"}
\newcommand{\BEpi}{"{\mathsf{BE}}_{\pi*}@\BEpi"}
\newcommand{\C}{"\mathsf{C}@\C"}
\newcommand{\D}{"\mathsf{D}@\D"}
\newcommand{\HS}{"\mathsf{HS}@\HS"}
\WithSuffix\newcommand\BE*{\mathsf{BE}}
\WithSuffix\newcommand\BEpi*{{\mathsf{BE}}_{\pi*}}
\WithSuffix\newcommand\C*{\mathsf{C}}
\WithSuffix\newcommand\D*{\mathsf{D}}
\WithSuffix\newcommand\HS*{\mathsf{HS}}

\newcommand\hsB{\mathop{"\langle B\rangle@\hsB"}}
\newcommand\hsE{\mathop{"\langle E\rangle@\hsE"}}
\newcommand\hsBu{\mathop{"[B]@\hsBu"}}
\newcommand\hsEu{\mathop{"[E]@\hsEu"}}
\newcommand\hsX{\mathop{"\langle X\rangle@\hsB"}}
\newcommand\hsG{\mathop{"\langle G\rangle@\hsG"}}
\newcommand\hsGu{\mathop{"[G]@\hsGu"}}
\newcommand\hsC{\mathbin{"\langle C\rangle@\hsB"}}
\newcommand\hsD{\mathop{"\langle D\rangle@\hsB"}}
\WithSuffix\newcommand\hsB*{\mathop{\langle B\rangle}}
\WithSuffix\newcommand\hsE*{\mathop{\langle E\rangle}}
\WithSuffix\newcommand\hsBu*{\mathop{[B]}}
\WithSuffix\newcommand\hsEu*{\mathop{[E]}}
\WithSuffix\newcommand\hsX*{\mathop{\langle X\rangle}}
\WithSuffix\newcommand\hsG*{\mathop{\langle G\rangle}}
\WithSuffix\newcommand\hsGu*{\mathop{[G]}}
\WithSuffix\newcommand\hsC*{\mathbin{\langle C\rangle}}
\WithSuffix\newcommand\hsD*{\mathop{\langle D\rangle}}

\newcommand\dec[1]{"\mathrm{flat}(#1)@\dec"}
\newcommand\enc[1]{"\mathrm{enc}(#1)@\enc"}
\WithSuffix\newcommand\dec*[1]{\mathrm{flat}(#1)}
\WithSuffix\newcommand\enc*[1]{\mathrm{enc}(#1)}

\let\oldmodels\models
\renewcommand\models{\mathbin{"\oldmodels@\models"}}
\WithSuffix\newcommand\models*{\mathbin{\oldmodels}}

\newcommand\true{"\mathrm{true}@\true"}
\WithSuffix\newcommand\true*{\mathrm{true}}

\newcommand\false{"\mathrm{false}@\false"}
\WithSuffix\newcommand\false*{\mathrm{false}}

\let\oldpi\pi
\renewcommand\pi{"\oldpi@\pi"}
\WithSuffix\newcommand\pi*{\oldpi}

\let\oldPhi\Phi
\renewcommand\Phi{"\oldPhi@\Phi"}
\WithSuffix\newcommand\Phi*{\oldPhi}

\newcommand\Args[1]{"{\boldsymbol\partial\mspace{-2mu}}^\star{#1}\mspace{2mu}@\Args"}
\newcommand\BArgs[1]{"{\boldsymbol\partial\mspace{-2mu}}_B{#1}\mspace{2mu}@\BArgs"}
\newcommand\EArgs[1]{"{\boldsymbol\partial\mspace{-2mu}}_E{#1}\mspace{2mu}@\EArgs"}
\newcommand\XArgs[1]{"{\boldsymbol\partial\mspace{-2mu}}_X{#1}\mspace{2mu}@\BArgs"}
\newcommand\YArgs[1]{"{\boldsymbol\partial\mspace{-2mu}}_Y{#1}\mspace{2mu}@\BArgs"}
\newcommand\Depth[2]{"\mathrm{Depth}^{\le#1}(#2)@\Depth"}
\WithSuffix\newcommand\Args*[1]{{\boldsymbol\partial\mspace{-2mu}}^\star{#1}\mspace{2mu}}
\WithSuffix\newcommand\BArgs*{{\boldsymbol\partial\mspace{-2mu}}_B\mspace{2mu}}
\WithSuffix\newcommand\EArgs*{{\boldsymbol\partial\mspace{-2mu}}_E\mspace{2mu}}
\WithSuffix\newcommand\XArgs*{{\boldsymbol\partial\mspace{-2mu}}_X\mspace{2mu}}
\WithSuffix\newcommand\YArgs*{{\boldsymbol\partial\mspace{-2mu}}_Y\mspace{2mu}}
\WithSuffix\newcommand\Depth*[2]{\mathrm{Depth}^{\le#1}(#2)}

\newcommand\type[2]{"\mathrm{type}^{#1}_{#2}@depth-$#1$ type"}
\WithSuffix\newcommand\type*[2]{\mathrm{type}^{#1}_{#2}}

\newcommand\lenone{"\texttt{1}@\lenone"}
\WithSuffix\newcommand\lenone*{\texttt{1}}
\newcommand\lentwo{"\texttt{2}@\lentwo"}
\WithSuffix\newcommand\lentwo*{\texttt{2}}
\newcommand\lenthree{"\texttt{3}@\lenthree"}
\WithSuffix\newcommand\lenthree*{\texttt{3}}

\let\oldcdot\cdot
\renewcommand\cdot[1]{\mathbin{"\oldcdot@depth-$#1$ composition"}}
\WithSuffix\newcommand\cdot*{\oldcdot}

\let\oldvdash\vdash
\renewcommand\vdash{\mathbin{"\oldvdash@\vdash"}}
\newcommand\vdashbis{\mathbin{"\oldvdash@\vdashbis"}}
\WithSuffix\newcommand\vdash*{\mathbin{\oldvdash}}
\WithSuffix\newcommand\vdashbis*{\mathbin{\oldvdash}}

\newcommand\dummy{"\emptystr@dummy"}
\WithSuffix\newcommand\dummy*{\emptystr}

\newcommand\comp[1]{"c(#1)@\comp"}
\WithSuffix\newcommand\comp*[1]{c(#1)}

\newcommand\leftDeltaIntervals[1]{"\mathrm{left}\text{-}\Delta_{#1}@\leftDeltaIntervals"}
\WithSuffix\newcommand\leftDeltaIntervals*[1]{\mathrm{left}\text{-}\Delta_{#1}}

\newcommand\rightDeltaIntervals[1]{"\mathrm{right}\text{-}\Delta_{#1}@\rightDeltaIntervals"}
\WithSuffix\newcommand\rightDeltaIntervals*[1]{\mathrm{right}\text{-}\Delta_{#1}}

\newcommand\encodedword[1]{"w_{#1}@encoding"}
\WithSuffix\newcommand\encodedword*[1]{w_{#1}}

\title{The Logic of Prefixes and Suffixes \\ 
       is Elementary under Homogeneity
       \thanks{ 
               We acknowledge the support 
               from the 2022 Italian INdAM-GNCS project 
               \emph{``Elaborazione del Linguaggio Naturale e Logica Temporale per la Formalizzazione di Testi''},
               ref.~no.~\texttt{CUP\_E55F22000270001}.
               We would also like to thank Alberto Molinari, Laura Bozzelli, and Adriano Peron for many useful discussions. Finally, Angelo Montanari would like to acknowledge the work done on the  problem, together with Gabriele Puppis and Pietro Sala, when he was on leave at LaBRI in Bordeaux. }
      }
       
\author{Dario Della Monica \\ University of Udine \\
        \url{dario.dellamonica@uniud.it} \and
        Angelo Montanari \\ University of Udine \\ 
        \url{angelo.montanari@uniud.it} \and
        Gabriele Puppis \\ University of Udine \\
        \url{gabriele.puppis@uniud.it} \and
        Pietro Sala \\ University of Versona \\
        \url{pietro.sala@univr.it}}
\date{}

\begin{document}

\maketitle
\sloppy

\begin{abstract}
%
In this paper, we study the finite satisfiability problem for the "logic $\BE*$" 
under the homogeneity assumption. $\BE$ is the cornerstone of Halpern and Shoham's 
interval temporal logic, and features modal operators corresponding to the prefix
(a.k.a.~``Begins'') and suffix (a.k.a.~``Ends'') relations on intervals. 
In terms of complexity, $\BE$ lies in between the "``Chop'' logic $\C*$", 
whose satisfiability problem is known to be non-elementary, and the 
$\pspace$-complete interval "logic $\D*$" of the sub-interval (a.k.a.~``During'') 
relation. $\BE$ was shown to be $\expspace$-hard, and the only known
satisfiability procedure is primitive recursive, but not elementary.
Our contribution consists of tightening the complexity bounds of the satisfiability 
problem for $\BE$, by proving it to be $\expspace$-complete. 
We do so by devising an equi-satisfiable normal form with boundedly many nested modalities. 
The normalization technique resembles Scott's quantifier elimination, but it turns 
out to be much more involved due to the limitations enforced by the homogeneity assumption.

\end{abstract}


\section{Introduction}\label{sec:intro}

\knowledge{notion}
    | ``Chop'' logic $\C*$
    | Chop logic $\C*$
    | logic $\C*$
	| \C
\knowledge{notion}
    | logic $\D*$
	| \D
\knowledge{notion}
    | logic $\HS*$
    | $\HS*$ logic
	| \HS
	| \HS*
\knowledge{notion}
	| model-checking problem
	| model-checking

In this paper, we study the computational complexity of the 
"satisfiability problem" for the "logic $\BE*$" of the 
"prefix" and "suffix" interval relations.
The considered interpretation setting is the one with intervals 
over finite linear orders, under the "homogeneity" assumption 
(see below). 
The "logic $\BE*$" is at the core of the galaxy of interval 
temporal logics \cite{hs91} and has interesting connections with standard 
point-based temporal logics \cite{DBLP:journals/tocl/BozzelliMMPS19}. 
In general, formulas of interval temporal logics can express 
properties of \emph{pairs} of time points, rather than 
properties of single time points, and are evaluated as 
sets of such pairs, that is, as binary relations on points. 
They are very expressive in comparison to point-based ones, 
and it does not come as a surprise that, in general, there 
is no reduction of their "satisfiability problem" to 
"satisfiability" of classical monadic second-order logic. 

\AP
The \reintro{logic $\BE*$} has two (unary) modalities, 
$\reintro{\hsB*}$ and $\reintro{\hsE*}$, 
that quantify over "prefixes" and "suffixes" of the 
current interval, respectively. 
These modalities can be viewed as the logical counterparts of 
Allen's binary relations \emph{Begins} and \emph{Ends}~\cite{all83}.
\AP
In particular, the "logic $\BE*$" can be considered as a fragment of 
Halpern and Shoham's interval temporal logic~\cite{hs91}, 
denoted $""\HS*""$, which features one modal operator for 
each of the twelve non-trivial Allen's relations.

The "satisfiability problem" for $\BE$ turns out to be 
\emph{undecidable} over all relevant classes of 
"interval structures"~\cite{dblp:journals/amai/bresolinmgms14,lod00}. 
One can however escape this bleak landscape by 
constraining the semantics, in particular, the 
interpretation of the propositional letters.
An interesting example is given by the \emph{"homogeneity"}
assumption, according to which a propositional letter 
"holds at" an interval if and only if it "holds at" 
all of its points. 
In other words, according to the "homogeneous" semantics, 
the labelling of an arbitrary interval in a model is 
uniquely determined by those of the singleton intervals 
contained in it.

An advantage of the "homogeneity" assumption is that
it makes it possible to define a natural interpretation of 
interval logics over Kripke structures. 
\AP
For example, this 
comes in handy when studying the ""model-checking problem"",
which is defined as the problem of deciding whether a given 
formula is "valid over" all ("homogeneous") "interval structures" 
generated by a given Kripke structure. 
As such, the "problem@model-checking" can be seen as a variant 
of the classical "validity"/"satisfiability problem", and many 
decidability and complexity results can be transferred 
from one problem to another.
In \cite{DBLP:journals/acta/MolinariMMPP16}
it was shown that the "model-checking" and "satisfiability" 
problems for $\BE$, and in fact for full $\HS$ logic,
become decidable when one restricts to 
\emph{homogeneous} "interval structures".
%

Despite its simple syntax and the "homogeneity" assumption,
the "logic $\BE*$" turns out to be quite expressive and succinct.
In~\cite{DBLP:journals/tocl/BozzelliMMPS19}, 
Bozzelli et al.~have shown that, when interpreted over 
finite words, LTL (Linear Temporal Logic) and $\BE$, 
under "homogeneity", define the same class of 
star-free regular languages, but with the latter formalism 
being at least exponentially more succinct than the former.
This is also reflected in the complexity of the 
"satisfiability problem" for $\BE$,
which was shown to be 
$\expspace$-hard~\cite[Theorem 3.1]{DBLP:journals/tcs/BozzelliMMPS19}.%
\footnote{In fact, the cited result focuses on the 
          "model-checking problem" for $\BE$,
          which takes as input, not only a formula, 
          but also a Kripke structure.
          It happens that the Kripke structure 
          used in the proof of the $\expspace$ 
          lowerbound generates every possible 
          "homogeneous" "interval structure", 
          and hence the result can be immediately
          transferred to the "validity"/"satisfiability problem" 
          for $\BE$.}
\AP
On the other hand, the only known decision procedure 
~\cite{DBLP:journals/acta/MolinariMMPP16}
for "satisfiability" of $\BE$ formulas 
is basically the one for full $\HS$, 
and is not elementary. 

\AP
It is also worth contrasting the expressiveness and 
complexity of $\BE$, under "homogeneity", with those 
of two close relatives of it: the 
"Chop logic $\C*$"~\cite{chopping_intervals}
and the "logic $\D*$" of the sub-interval relation~\cite{LMCS2022}.
\AP
The ""logic $\C*$"" has a binary modality $\hsC*$ 
that allows one to split the current interval 
in two parts and predicate separately on them. 
\AP
The ""logic $\D*$"" has a unary modality $\hsD*$ 
that allows one to predicate 
about sub-intervals of the current interval. 
It is easy to see that, in terms of expressiveness, 
$\BE$ lies in between $\D$ and $\C$, in the sense 
that modality $\hsD$ can be defined 
in $\BE$, i.e.,~$\hsD\varphi$ 
is "equivalent" to $\hsB\hsE\varphi$,
and modalities $\hsB$ and $\hsE$ can in turn be defined 
in $\C$, e.g.,~$\hsB\varphi$ is "equivalent" to $\varphi \mathop{\hsC} \true$.
Under the "homogeneity" assumption, 
the "satisfiability problem" for $\C$ is non-elementarily 
decidable, precisely, tower-complete, in view of the existence of 
straightforward reductions to and from language-emptiness of generalized 
star-free regular expressions~\cite{Schmitz:2016,stockmeyer1974},
while the "satisfiability problem" for $\D$ 
was shown to be $\pspace$-complete by a suitable 
contraction method~\cite{LMCS2022}.
It is also worth pointing out that if the "homogeneity" 
assumption is removed, the "satisfiability problem" for 
$\D$ becomes undecidable~\cite{DBLP:journals/fuin/MarcinkowskiM14}.


Based on the observations above,
it is crucial
to close the gap between the complexity lowerbound and 
upperbound of the "satisfiability problem" for $\BE$.
Significant effort has been invested in recent years 
towards both raising the $\expspace$ lowerbound,
e.g.~using variants of Stockmeyer's counters~\cite{stockmeyer1974},
and developing elementary "satisfiability" procedures.
Despite these efforts, the complexity gap remained unchanged
and proved to be an intriguing challenge. 
The special status of $\BE$ is witnessed by the fact that 
the many results about the complexity of the "satisfiability" 
and/or "model-checking" problems for proper fragments of $\HS$, 
under the "homogeneity" assumption, concern logics that include 
neither modality $\reintro{\hsB*}$ nor modality $\reintro{\hsE*}$ 
or feature only one of them (an up-to-date picture can be found in \cite{DBLP:journals/corr/abs-2109-08320,DBLP:conf/time/BozzelliMPS21}).

In this paper, we manage to prove that the "satisfiability problem" 
for the logic $\BE$, under "homogeneity", is elementarily decidable, 
and precisely $\expspace$-complete.
This result is established using a rather unexpected 
normalization technique, which consists of transforming 
an arbitrary $\BE$ formula into an "equi-satisfiable" one 
with boundedly many nested modalities. 
Specifically, we will show that one can compute, 
in polynomial time, "normalized@shallow normal form" 
formulas with nesting "depth" of modalities at most 4, 
and with at most 2 alternations between universal and 
existential modalities.
The transformation of $\BE$ formulas into "normalized" 
ones can be also viewed as a quantifier elimination 
technique à-la Scott~\cite{Scott1962}.
In this perspective, however, the transformation 
has to deal with an increased difficulty: 
due to the "homogeneity" assumption, the elements 
over which we predicate cannot be labelled in an 
arbitrary way. 
In view of this difficulty, it is quite surprising 
that an "equi-satisfiable" "normalized"
formula can be computed in polynomial time from any given 
arbitrary $\BE$ formula.

The rest of the paper is organized as follows. 
In Section \ref{sec:prelims}, we introduce 
the logic $\BE$ and we point out the relevant 
implications of the "homogeneity" assumption.  
In Section \ref{sec:translation}, we define the 
transformation of $\BE$ formulas into "normalized" ones.  
In Section \ref{sec:complexity}, we derive an 
optimal "satisfiability" procedure and analyse 
its complexity.
Conclusions provide an assessment of the work done 
and outline future research directions.
For reader convenience, technical terms and notation 
in the electronic version of the paper 
are linked to their definitions, which can then be 
accessed with a mouse click.


\section{Preliminaries}\label{sec:prelims}

\knowledge{notion}
	| time domain
	| (\timedomain*,<)
\knowledge{notion}
	| interval domain 
	| \intervaldomain*
\knowledge{notion}
	| interval structure
	| interval structures
	| structure
	| structures
\knowledge{notion}
	| prefix relation
	| proper prefix
	| proper prefixes
	| prefix
	| prefixes
\knowledge{notion}
	| suffix relation
	| proper suffix
	| proper suffixes
	| suffix
	| suffixes
\knowledge{notion}
	| signature
	| signatures
\knowledge{notion}
    | logic $\BE*$
	| \BE
\knowledge{notion}
    | logic $\BEpi*$
	| \BEpi
\knowledge{notion}
	| \hsB
	| \hsB*
\knowledge{notion}
	| \hsE
	| \hsE*
\knowledge{notion}
	| \hsBu
\knowledge{notion}
	| \hsEu
\knowledge{notion}
	| \models
	| interpreted
	| satisfies
	| satisfy
	| satisfied by
	| satisfying
	| holds at
	| hold at
	| held at
	| model
	| models
	| modelled
	| evaluates
	| evaluated
	| evaluation
	| evaluations
\knowledge{notion}
	| valid
	| validity
\knowledge{notion}
	| valid over
	| validity over
\knowledge{notion}
	| satisfiable
	| satisfiability
	| satisfiability problem
\knowledge{notion}
	| equivalent
	| equivalence
\knowledge{notion}
	| equi-satisfiable
	| equi-satisfiability
\knowledge{notion}
	| \true
\knowledge{notion}
	| \false
\knowledge{notion}
	| \pi
\knowledge{notion}
	| \hsGu
\knowledge{notion}
	| homogeneous
	| homogeneity
	| Homogeneity
\knowledge{notion}
	| homogeneous normal form

\AP
Let the ""time domain"" be a finite prefix of the natural numbers $"(\timedomain*,<)"$.
\AP
Intervals over $\timedomain$ are denoted by $[x,y]$, for $x,y\in \timedomain$ and $x\le y$,
and the set of all intervals over $\timedomain$ is denoted 
$""\intervaldomain*""$. 
\AP
We let $\prefix$ (resp., $\suffix$) be the proper ""prefix""
(resp., ""suffix"") relation on intervals, 
defined by $J \prefix I$ if and only if $\min(I)=\min(J)\le\max(J)<\max(I)$ 
(resp., $J \suffix I$ if and only if $\min(I)<\min(J)\le\max(J)=\max(I)$). 

\AP
Formulas of the ""logic $\BE*$"" are constructed starting from 
propositional letters belonging to a finite non-empty set $\signature$, 
called ""signature"", using classical Boolean connectives and modal 
operators. The latter operators are used to quantify over "prefixes" 
and "suffixes" of the current interval.
\AP
Formally, $\BE$ formulas satisfy the following grammar:
\[
  \varphi ~~::=~~
  p ~~(\text{for }p\in\signature) ~~|~~ 
  \neg\varphi ~~|~~ 
  \varphi \vee \varphi ~~|~~ 
  \hsB \varphi ~~|~~ 
  \hsE \varphi. 
\]
Semantics is given in terms of an "interval structure" $\cS$ and one of its intervals $I$. 
\AP
Formally, an ""interval structure"" over a "signature" $\signature$ is a pair 
$\cS=\intervalstructure$, where 
$\sigma:\intervaldomain \rightarrow \wp(\signature)$ is a labelling
of intervals by subsets of $\signature$.
\AP
Whether a $\BE$ formula $\varphi$ ""holds at"" an interval $I$ of $\cS$, 
denoted $\cS,I\models\varphi$, is determined by the following rules:
\begin{itemize}
  \item $\cS,I \models p$ if $p\in\sigma(I)$;
  \item $\cS,I \models \neg \varphi$ if $\cS,I \not\models \varphi$; 
  \item $\cS,I \models \varphi_1 \vee \varphi_2$ if
        $\cS,I \models \varphi_1$ or $\cS,I \models \varphi_2$;
  \item \AP $\cS,I \models \mathop{""\hsB*""} \varphi$ if $\cS,J \models \varphi$ for some $J \prefix I$;
  \item \AP $\cS,I \models \mathop{""\hsE*""} \varphi$ if $\cS,J \models \varphi$ for some $J \suffix I$.
\end{itemize}
\AP
A formula is ""valid"" if it holds at every interval of every "interval structure"; 
\AP
similarly, it is ""satisfiable"" if it holds at some interval of some "interval structure".  
\AP
Two formulas $\varphi$ and $\varphi'$ are ""equivalent"" if for every interval 
structure $\cS$ and every interval $I$ in it, $\cS,I\models\varphi$ iff $\cS,I\models\varphi'$.
\AP
They are ""equi-satisfiable"" if either they are both "satisfiable" or none of them is.
\AP
The notions of "validity", "satisfiability", and "equivalence" can be relativized to a 
specific class of "interval structures" (possibly even to a single "interval structure").
\AP
As an example, we say that a formula $\varphi$ is ""valid over"" a class $\sC$ of 
"interval structures" if $\cS,I\models \varphi$ for all $\cS\in\sC$ and all 
$I\in\cS$. In the particular case where $\sC$ contains a single "interval structure" $\cS$, 
we will say that a formula $\varphi$ is \reintro{valid over} $\cS$ if $\cS,I\models\varphi$
for all $I\in\cS$.

It is  possible to add syntactic sugar to the logic $\BE$. 
\AP
As an example, we will often use shorthands like
$\varphi_1 \wedge \varphi_2 = \neg(\neg\varphi_1 \vee \neg\varphi_2)$,
$""\false*@\false"" = p\wedge \neg p$ (for any $p\in\signature$),
$""\true*@\true"" = \neg\false$,
and $""[X]@\hsBu""\varphi ""@\hsEu"" = \neg\langle X\rangle\neg\varphi$, for $X\in\{B,E\}$.
\AP
Some other useful shorthands are $""\pi*@\pi"" = \hsBu \false$, 
which constrains the interval where it is evaluated to be a singleton, and 
$""\hsGu*@\hsGu"" \varphi = \varphi \:\wedge\: \hsBu\varphi \:\wedge\: \hsEu\varphi \:\wedge\: \hsBu\hsEu\varphi$,
which constrains all sub-intervals (including the current interval, its "proper prefixes", and 
its "proper suffixes") to "satisfy" $\varphi$. 
The shorthands $\pi$ and $\hsGu$ can be viewed as derived nullary 
and unary modal operators, respectively, and can be added as syntactic 
sugar to  $\BE$.

\paragraph{Homogeneity assumption}
We recall from \cite{lod00,DBLP:journals/acta/MolinariMMPP16} 
that the "satisfiability" problems for the logic $\BE$ is 
undecidable, unless one restricts to "homogeneous" "interval structures".
\AP
An interval structure $\cS=\intervalstructure$ is ""homogeneous"" if its labelling
satisfies the condition $\sigma(I) = \bigcap_{x\in I}\sigma([x,x])$ for all $I\in\intervaldomain$.
Intuitively, this means that the labelling $\sigma$ is uniquely determined 
by its restriction to singleton intervals.
Let us take a closer look at the implications of "homogeneity".

\AP
First of all, we have that every formula $\hsB (p_1 \wedge p_2)$ 
is "equivalent" to $\hsB p_1 \wedge \hsB p_2$, and similarly for $\hsE$. 
Note, however, that "homogeneity" does not imply
similar properties for arbitrary formulas $\varphi_1,\varphi_2$ 
replacing the propositional letters $p_1,p_2$.
As an example, the formulas $\hsB(p \wedge \neg p)$ and
$(\hsB p) \wedge (\hsB\neg p)$ are not "equivalent".

\AP
"Homogeneity" can also be exploited to efficiently rewrite 
any $\BE$ formula into an "equivalent" one where every occurrence 
of a propositional letter is conjoined with $\pi$. 
Based on this observation, we introduce the following mild 
"normal form@homogeneous normal form":

\begin{definition}\label{def:homogeneus-normal-form}
A $\BE$ formula $\psi$ is in ""homogeneous normal form"" 
if every occurrence of a propositional letter $p$ in $\psi$ 
appears inside the subformula $\pi\wedge p$.
\end{definition}

Basically, the "homogeneous normal form" 
restricts propositional letters to be only 
evaluated at singleton intervals. 
As an example, the formula
$(\pi\wedge q) \vee \hsB(\pi \wedge \neg(\pi\wedge p))$ 
is in "homogeneous normal form" and "holds at" an 
interval $I$ iff $I$ consists of a single point 
labelled by $q$ or the left endpoint of $I$ is
not labelled by $p$.

\begin{proposition}\label{prop:homogeneus-normal-form}
One can transform in linear time any formula $\psi$ 
into one in "homogeneous normal form" that is "equivalent" 
to $\psi$ when interpreted over \emph{"homogeneous"}
"interval structures".
\end{proposition}

\begin{proof}
It suffices to replace 
every occurrence of a propositional letter $p$ in $\psi$ by the formula
$\mathtt{everywhere}(p) = 
 (\pi \wedge p) \:\vee\: 
 \big(\hsB(\pi \wedge p) \:\wedge\: \hsE(\pi \wedge p) \:\wedge\:
      \hsBu(\pi \vee \hsE(\pi\wedge p))\big)$.
The resulting formula is "equivalent" to $\psi$ since,
over "homogeneous" "interval structures", $p$ is "equivalent" to $\mathtt{everywhere}(p)$.
\end{proof}

\AP
We denote by $""\BEpi*@\BEpi""$ the fragment of "logic $\BE*$"
that contains only formulas in "homogeneous normal form".
\emph{From this point forward, we will exclusively work with 
$\BEpi$ formulas, with the understanding that this assumption 
may occasionally go unstated.
Accordingly, we will treat (sub)formulas of the form 
$\pi\wedge p$ as atomic.}


\section{A bounded-nesting normal form for \texorpdfstring{$\BEpi$}{BEpi}}\label{sec:translation}

\knowledge{notion}
	| modal depth
	| depth
	| depths
\knowledge{notion}
	| shallow normal form
	| normalized
	| normalization

In this section, we describe a transformation of arbitrary $\BEpi$ formulas 
into "equi-satisfiable" ones with boundedly many nested modalities.
The transformation is somehow reminiscent of the so-called 
Scott normal form
for the two-variable fragment of first-order logic \cite{Scott1962}, 
since it results in a formula, over an extended set of propositional letters, 
that is "satisfiable" if and only if the original formula was. 
The increased difficulty here is that the valuation of the new 
propositional letters emerging from the transformation 
must satisfy the "homogeneity" assumption.
This is to say that we cannot identify intervals 
"satisfying" a certain (sub)formula $\varphi$
by labelling them with a fresh propositional letter $q_\varphi$.
Rather, we will identify these intervals by appropriately 
correlating fresh labels assigned to their endpoints.
Our transformation will exploit in a crucial way the fact that,
under "homogeneity", valuations of formulas at two overlapping intervals 
have ``less degrees of freedom'' than valuations of the same formulas
at disjoint intervals.

\begin{definition}\label{def:modal-depth-and-normal-form}
The ""modal depth"" (or simply \reintro{depth}) of a $\BEpi$ formula 
is the maximum number of nested modal operators $\hsB$ and $\hsE$ in it, 
not counting those defining the operator $\pi$.
\AP
A $\BEpi$ formula is in ""shallow normal form"" if it is of the form
$\psi \:\wedge\: \hsGu\xi$, where both $\psi$ and $\xi$
have "depth" at most $2$.
\end{definition}

Concerning the above definition, 
we recall that $\hsGu\xi$ is a shorthand for 
$\xi \:\wedge\: \hsBu\xi \:\wedge\: \hsEu\xi \:\wedge\: \hsBu\hsEu\xi$,
so a formula in "shallow normal form" has "depth" at most $4$.
However, not all "depth"-$4$ formulas are in "shallow normal form".

\begin{theorem}\label{thm:main}
Given any $\BEpi$ formula $\psi$, one can compute in
in polynomial time an "equi-satisfiable" formula 
$\psi^\star$ that is in "shallow normal form".
\end{theorem}

To highlight one of the key ideas underlying the proof of the theorem, 
which we postpone to the next subsections, we give an example of 
normalization of a formula. 

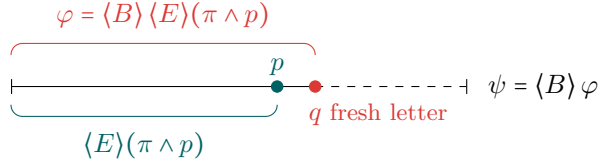
\begin{figure}[t]
\centering
\begin{tikzpicture}
\draw [|-] (0,0) -- (4,0);
\draw [-|, dashed] (4,0) -- (6,0);
\draw (7,0) node {$\psi=\hsB\varphi$};
\draw [nicered, above round brace, yshift=12] (0,0) -- (4,0)
      node [above, midway, yshift=10] 
      {$\varphi = \hsB*\hsE*(\pi*\wedge p)$};
\draw [nicecyan, below round brace, yshift=-7] (0,0) -- (3.5,0)
      node [below, midway, yshift=-10] 
      {$\hsE*(\pi*\wedge p)$};
\draw [nicered] (4,0) node [dot] {};
\draw [nicered] (4,0) node [below=2mm] 
      {$q \text{\small\rightward{ fresh letter}}$};
\draw [nicecyan] (3.5,0) node [dot] {};
\draw [nicecyan] (3.5,0) node [above=1.5mm] {$p$};
\end{tikzpicture}
\caption{Example of normalization of a formula $\psi=\hsB\varphi$.}
\label{fig:example}
\end{figure}

\begin{example}\label{ex:normalization}
Consider the formula $\psi=\hsB\varphi$ 
over the "signature" $\signature=\{p\}$, 
where $\varphi=\hsB\hsE (\pi \wedge p)$.
Figure \ref{fig:example} shows an example of an "interval structure" "satisfying" $\psi$; in particular, it highlights intervals 
witnessing $\varphi$ (in red) and $\hsE(\pi\wedge p)$ (in blue).
Note that $\psi$ has "depth" $3$ and is not in "shallow normal form". 
To rewrite $\psi$ into an "equi-satisfiable" formula in "shallow normal form",
we introduce a new propositional letter $q$ with the purpose of 
marking the right endpoints of the intervals that 
"satisfy" $\varphi$ and that are minimal w.r.t.~the "prefix" relation 
(we call these intervals prefix-minimal, for short).
Note that the right endpoints of these intervals
are immediately to the right of the $p$-labelled points.
We thus consider "interval structures" over the expanded 
"signature" $\signaturebis=\{p,q\}$ that make the following 
formula "valid@valid over":
\[
  \xi ~=~
  \underbrace{\big(\neg \pi \:\wedge\: \neg\hsB\neg\pi\big)}_%
             {\midward{\text{interval has exactly two points}}}
  ~~~~\then~~~~
  \underbrace{\big(\hsB(\pi\wedge p) \:\iff\: \hsE(\pi\wedge q)\big)}_ %
             {\midward{\text{$q$ is to the right whenever $p$ is to the left}}} 
\]
We can verify that, over "interval structures" 
that make $\xi$ "valid@valid over",
every prefix-minimal interval that 
"satisfies" $\varphi$ also "satisfies"
$\varphi' = \hsE(\neg\pi) \:\wedge\: \hsE (\pi \wedge q)$,
and, conversely, every interval that "satisfies" $\varphi'$ 
also "satisfies" $\varphi$.
This implies that, again over "interval structures" 
that make $\xi$ "valid@valid over", 
the "depth"-$3$ formula $\psi = \hsB\varphi$ 
is "equivalent" to the "depth"-$2$ 
formula $\psi' = \hsB\varphi'$.
Moreover, since the labelling of any "interval structure" 
over $\signature=\{p\}$ 
can always be expanded with the fresh letter $q$ 
so as to "satisfy" $\hsGu\xi$, we conclude 
that $\psi$ is "equi-satisfiable" as the formula 
$\psi^\star = \psi' \:\wedge\: \hsGu\xi$.
Since $\xi$ has "depth" $1$, $\psi^\star$ is also in "shallow normal form".
\end{example}

The normalization procedure for an arbitrary formula $\psi$
iterates a rewriting similar to the one presented in 
Example~\ref{ex:normalization}.
More precisely, we start by replacing every outermost subformula of $\psi$ of 
"depth" $d>2$ and of the form $\hsB\varphi$ (resp., $\hsE\varphi$) with an 
equi-satisfiable formula $\hsB\varphi'$ (resp., $\hsE\varphi'$) of "depth" $2$.
This rewriting step extends the "signature" with new propositional letters,
which are constrained while preserving "equi-satisfiability" using 
formulas similar to the $\hsGu\xi$ of Example \ref{ex:normalization}. 
Constraints will contain occurrences of the original subformula 
$\varphi$, and thus need to be normalized in their turn in order to eventually 
obtain formulas of "depth" at most $2$. 
More details and formal arguments about the normalization procedure of Theorem 
\ref{thm:main} will be provided in the next subsections.

We conclude this part by observing an immediate consequence of Theorem \ref{thm:main}.
We recall from \cite{DBLP:journals/acta/MolinariMMPP16} the existence of a rather simple,
but non-elementary procedure for deciding "satisfiability" of 
a $\BE$ formula $\psi$ under "homogeneity".
A close inspection to the description of this procedure shows that it has 
non-deterministic time complexity $\cO(\mathit{tow}(h,|\psi|))$, where 
$\mathit{tow}(h,n) = 2^{2^{\iddots^{^n}}}$ is the tower 
of $h$ exponents ending with $n$ and $h$ is the maximum 
number of nested modal operators in the input formula $\psi$.
As the shorthand $\pi$ can be directy handled in constant time, 
the parameter $h$ of the said complexity bound 
can be identified with
our notion of "modal depth" for $\BEpi$ formulas.
In particular, when we consider a formula $\psi$ 
in "shallow normal form", the parameter $h$ is at most $4$. 
Together with Proposition \ref{prop:homogeneus-normal-form}
and Theorem \ref{thm:main}, this gives a first rough complexity 
bound to the satisfiability problem for $\BE$ logic 
under the "homogeneity" assumption:

\begin{corollary}\label{cor:main}
The "satisfiability problem" for $\BE$ logic restricted
to "homogeneous" "interval structures"
is elementarily decidable, i.e.,~at least in $4\nexptime$.
\end{corollary}

We shall provide later, in Section \ref{sec:complexity}, 
a more careful complexity analysis, showing that the 
"satisfiability problem" for $\BE$ logic under "homogeneity"
is actually $\expspace$-complete.

      
\knowledge{notion}
	| expansion
	| expanded
\knowledge{notion}
	| expander
	| expanders

\subsection{Expanders}\label{subsec:expanders}

A first ingredient of the "normalization" procedure 
of $\BEpi$ formulas is 
that of an "expander". Intuitively, this is a formula that constrains 
new propositional letters on the basis of the old ones in an arbitrary 
("homogeneous") "interval structure". 

\begin{definition}\label{def:expander}
Let $\signature\subseteq\signaturebis$ be two "signatures", and 
let $\cS=\intervalstructure$ and $\cS'=\intervalstructurebis$ be 
"interval structures" over $\signature$ and $\signaturebis$, respectively.
We say that $\cS'$ is an ""expansion"" of $\cS$ if $\timedomainbis=\timedomain$ and 
$\sigma'(I)\cap\signature=\sigma(I)$ for all intervals $I\in\intervaldomain$.

\AP
An ""expander"" from $\signature$ to $\signaturebis$
is a $\BEpi$ formula $\xi$ over $\signaturebis$ such that,
for every "interval structure" $\cS$ over $\signature$, 
there is an "expansion" $\cS'$ of $\cS$ over $\signaturebis$ 
that makes $\xi$ "valid@valid over".
\end{definition}

We report below a simple lemma about expanders.

\begin{restatable}{lemma}{ExpanderEquivalence}
\label{lem:expander-equivalence}
If $\xi$ is an "expander" from $\signature$ to $\signaturebis$,
$\psi$ and $\psi'$ are formulas over the "signatures" 
$\signature$ and $\signaturebis$, respectively, and
$\psi,\psi'$ are "equivalent" over all "interval structures" 
where $\xi$ is "valid@valid over",
then $\psi$ and $\psi' \wedge \hsGu\xi$ are "equi-satisfiable".
\end{restatable}

\begin{proof}
Suppose that $\psi' \wedge \hsGu\xi$ is "satisfied by" 
an "interval structure" $\cS'$ over $\signaturebis$.
Because, $\xi$ is "valid over" $\cS'$, $\psi$ is
"equivalent" to $\psi'$ over $\cS'$, and hence 
$\cS'$ "satisfies" $\psi$.
Conversely, if $\psi$ is "satisfied by" an "interval structure" 
$\cS$ over $\signature$, then there is an "expansion"
$\cS'$ of $\cS$ that makes 
$\xi$ "valid@valid over". 
This implies that $\psi$ and $\psi'$ are "equivalent"
over $\cS'$. Hence $\cS'$ "satisfies" $\psi'$, and
$\psi' \wedge \hsGu\xi$ as well.
\end{proof}

\knowledge{notion}
	| prefix-minimal
	| prefix-minimality
	| minimal
\knowledge{notion}
	| suffix-minimal
	| suffix-minimality

\subsection{Minimal witnessing intervals}\label{subsec:minimal-intervals}

Recall that the "normalization" of a $\BEpi$ formula  
replaces subformulas $\hsB\varphi$ (resp., $\hsE\varphi$) 
of "depth" $d>2$ with equivalent formulas $\hsB\varphi'$ (resp., $\hsE\varphi'$) 
of "depth" $2$.
In this respect, a simple observation is that,
in order to determine which intervals "satisfy" 
$\hsB\varphi$ (resp., $\hsE\varphi$), one could 
look at intervals that satisfy $\varphi$ and that
are \emph{minimal} for the "prefix" (resp., "suffix")
relation.

\begin{definition}\label{def:minimal}
\AP
Given a $\BEpi$ formula $\varphi$, 
an "interval structure" $\cS$, and an interval $I$ in it, 
we say that $I$ is ""prefix-minimal"" (resp., ""suffix-minimal"") 
for $\varphi$ if $\cS,I\models\varphi$ and $\cS,J\not\models\varphi$ 
for every $J \prefix I$ (resp., $J \suffix I$). 
\end{definition}

We will see later that "prefix@prefix-minimal"/"suffix-minimal" 
intervals for $\varphi$ can be unambiguously identified, 
once their endpoints are annotated with fresh propositional letters, 
using a formula $\varphi'$ of size proportional 
to that of $\varphi$, but with "depth" just $1$.
A simplified account of this technique was already given in Example \ref{ex:normalization}.
Below, we discuss the approach under a more general 
perspective and highlight a potential issue with 
overlapping "minimal" witnesses.

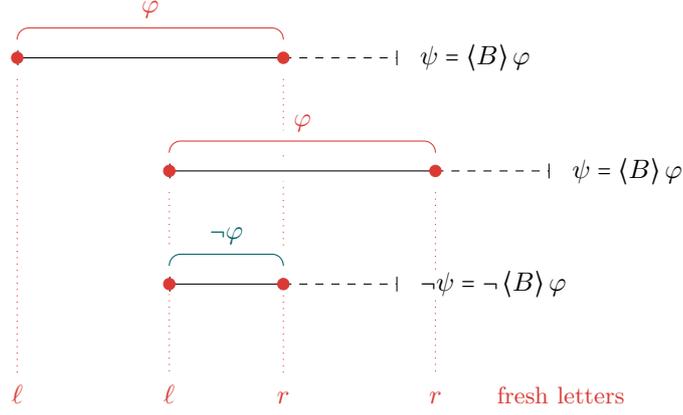
\begin{figure}[t]
\centering
\begin{tikzpicture}
\coordinate (A) at (0,0);
\coordinate (A') at (0,-4.5);
\coordinate (B) at (3.5,0);
\coordinate (B') at (3.5,-4.5);
\coordinate (C) at (5,0);
\coordinate (D) at (2,-1.5);
\coordinate (D') at (2,-4.5);
\coordinate (E) at (5.5,-1.5);
\coordinate (E') at (5.5,-4.5);
\coordinate (F) at (7,-1.5);
\coordinate (G) at (2,-3);
\coordinate (G') at (2,-4.5);
\coordinate (H) at (3.5,-3);
\coordinate (H') at (3.5,-4.5);
\coordinate (I) at (5,-3);
\draw [dotted, nicered, short] (A) -- (A');
\draw [dotted, nicered, short] (B) -- (B');
\draw [dotted, nicered, short] (D) -- (D');
\draw [dotted, nicered, short] (E) -- (E');
\draw [white, line width=8pt] (A) -- (C);
\draw [white, line width=8pt] (D) -- (F);
\draw [white, line width=8pt] (G) -- ([xshift=20] I);
\draw [|-] (A) -- (B);
\draw [-|, dashed] (B) -- (C);
\draw (C) node [right=3mm] {$\psi=\hsB\varphi$};
\draw [nicered, above round brace] ([yshift=7] A) -- ([yshift=7] B)
      node [above, midway, yshift=8] {$\varphi$};
\draw [nicered] (A) node [dot] {};
\draw [nicered] (B) node [dot] {};
\draw [|-] (D) -- (E);
\draw [-|, dashed] (E) -- (F);
\draw (F) node [right=3mm] {$\psi=\hsB\varphi$};
\begin{scope}
\draw [white, line width=8pt, above round brace] 
      ([yshift=7] D) -- ([yshift=7] E);
\end{scope}
\draw [nicered, above round brace] ([yshift=7] D) -- ([yshift=7] E)
      node [above, midway, yshift=8] {$\varphi$};
\draw [nicered] (D) node [dot] {};
\draw [nicered] (E) node [dot] {};
\draw [|-] (G) -- (H);
\draw [-|, dashed] (H) -- (I);
\draw (I) node [fill=white, rectangle, right=3mm] {$\neg\psi=\neg\hsB\varphi$};
\begin{scope}
\draw [white, line width=8pt, above round brace] 
      ([yshift=7] G) -- ([yshift=7] H);
\end{scope}
\draw [nicecyan, above round brace] ([yshift=7] G) -- ([yshift=7] H)
      node [above, midway, yshift=8] {$\neg\varphi$};
\draw [nicered] (G) node [dot] {};
\draw [nicered] (H) node [dot] {};
\draw [nicered] (A') node {$\ell$}; 
\draw [nicered] (D') node {$\ell$}; 
\draw [nicered] (B') node {$r$}; 
\draw [nicered] (E') node {$r$}; 
\draw [nicered] ([xshift=20] E') node 
      {\text{\small\rightward{ fresh letters}}};
\end{tikzpicture}
\caption{Overlapping "prefix-minimal" intervals for $\varphi$, 
and their intersection.}
\label{fig:overlapping-witnesses}
\end{figure}

\begin{example}\label{ex:overlapping}
Suppose that $\varphi$ is a formula of "depth" $2$.
We aim at replacing it with a formula $\varphi'$ of "depth" $1$, 
so that $\hsB\varphi'$ turns out to be "equivalent" to $\hsB\varphi$ 
in an appropriate "expansion" of the "interval structure".
As discussed earlier, a natural approach is to focus only on intervals 
that are "prefix-minimal" for $\varphi$, and mark their endpoints with 
suitable fresh propositional letters.
For example, two "prefix-minimal" intervals for $\varphi$
are represented in Figure \ref{fig:overlapping-witnesses}
by the red brackets.
We mark their left and right endpoints with fresh 
propositional letters $\ell$ and $r$, respectively,
and we assume that the "interval structure" is 
"expanded" so as to satisfy the intended use of 
$\ell$ and $r$. 
We then define 
$\varphi' = (\hsB(\pi \wedge \ell) \wedge \hsE(\pi \wedge r)) \vee 
            ((\pi \wedge \ell) \wedge (\pi \wedge r))$
and observe that every interval "satisfying" $\hsB\varphi$ 
must also "satisfy" $\hsB\varphi'$.
So one might be tempted to replace $\varphi$ with $\varphi'$.
Unfortunately, while $\hsB\varphi$ entails $\hsB\varphi'$, 
the converse is not true,
as the intersection of any two "prefix-minimal" intervals 
for $\varphi$ does not always "satisfy" $\varphi$ 
(see the blue bracket in Figure \ref{fig:overlapping-witnesses}).
In general, in order to mark the endpoints of 
"minimal" intervals without ambiguities, 
one could use different letters to mark the endpoints
of any two overlapping intervals. 
More precisely, one should 
introduce as many copies of letters 
$\ell,r$ as the maximum number of overlapping "prefix-minimal" 
intervals for $\varphi$ 
that have different right endpoints.
\end{example}


               
\knowledge{notion}
    | intersecting
	| intersecting family
	| intersecting families
\knowledge{notion}
	| chain
	| chains
\knowledge{notion}
	| anti-chain
	| anti-chains
\knowledge{notion}
	| \BArgs
	| \BArgs*\varphi
\knowledge{notion}
	| \EArgs
	| \EArgs*\varphi
\knowledge{notion}
	| $\varphi$-profile
	| $\varphi$-profiles
	| profile
	| profiles
\knowledge{notion}
	| \leftDeltaIntervals
	| \leftDeltaIntervals*
	| \leftDeltaIntervals{i}
	| \leftDeltaIntervals*{i}
\knowledge{notion}
	| \rightDeltaIntervals
	| \rightDeltaIntervals*
	| \rightDeltaIntervals{i}
	| \rightDeltaIntervals*{i}
\knowledge{notion}
	| special $\hsB*$-witness of $\alpha$ at $i$
	| special $\hsB*$-witness
	| special $\hsB*$-witnesses
	| special $\hsX*$-witness
	| special $\hsX*$-witnesses
	| special witness
	| special witnesses
	| Special witnesses
\knowledge{notion}
	| special $\hsE*$-witness of $\alpha$ at $i$
	| special $\hsE*$-witness
	| special $\hsE*$-witnesses
	| Special $\hsE*$-witnesses
\knowledge{notion}
    | \dec
    | \dec*{\varphi}
\knowledge{notion}
    | \enc
    | \enc*{\varphi}

\subsection{Encoding of minimal witnessing intervals}
\label{subsec:intersecting-families}

Example \ref{ex:overlapping} brings up a third ingredient 
that is crucial for the "normalization" procedure,
as it suggests that, in order to mark without ambiguities 
the endpoints of "prefix-minimal" 
(resp., "suffix-minimal") intervals for a formula $\varphi$,
one must first bound the number of distinct right (resp., left) 
endpoints of overlapping intervals. 
A bound will be shown precisely in 
Corollary~\ref{cor:core-bound-i} below.

\begin{definition}\label{def:intersecting-family}
A set $\cI$ of intervals is an ""intersecting family"" 
if there is a point $x$ that is contained in every interval of $\cI$.
\end{definition}

An example of an "intersecting family" of intervals
is shown to the left of Figure \ref{fig:chain-antichain}.

Towards proving the desired bound,
we shall first establish two auxiliary lemmas.
The first lemma relates the maximum 
cardinality of a partially ordered set
(e.g., an "intersecting family" of intervals,
partially ordered by containment) to the maximum 
cardinality of its "chains" and "anti-chains".
\AP
Formally, a ""chain"" of a partially ordered set 
is a subset of pairwise comparable elements.
\AP
An ""anti-chain"" is a subset of pairwise 
incomparable elements. 
The first lemma is in fact a rephrasing of 
Dilworth's theorem~\cite{dilworth50} 
(we give a proof here for self-containment):

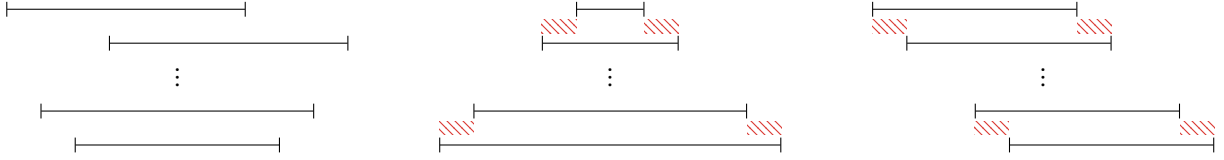
\begin{figure*}[t]
\centering
\begin{tikzpicture}[scale=0.9]
\begin{scope}[xshift=0]
\draw [|-|] (-2.5,0) -- (1,0);
\draw [|-|] (-1,-0.5) -- (2.5,-0.5);
\draw (0,-1) node {$\vdots$};
\draw [|-|] (-2,-1.5) -- (2,-1.5);
\draw [|-|] (-1.5,-2) -- (1.5,-2);
\end{scope}
\begin{scope}[xshift=180]
\draw [|-|] (-0.5,0) -- (0.5,0);
\fill [pattern color=nicered, pattern=north west lines] 
      (-0.5,-0.15) rectangle (-1,-0.35);
\fill [pattern color=nicered, pattern=north west lines] 
      (0.5,-0.15) rectangle (1,-0.35);
\draw [|-|] (-1,-0.5) -- (1,-0.5);
\draw (0,-1) node {$\vdots$};
\draw [|-|] (-2,-1.5) -- (2,-1.5);
\fill [pattern color=nicered, pattern=north west lines] 
      (-2,-1.65) rectangle (-2.5,-1.85);
\fill [pattern color=nicered, pattern=north west lines] 
      (2,-1.65) rectangle (2.5,-1.85);
\draw [|-|] (-2.5,-2) -- (2.5,-2);
\end{scope}
\begin{scope}[xshift=360]
\draw [|-|] (-2.5,0) -- (0.5,0);
\fill [pattern color=nicered, pattern=north west lines] 
      (-2.5,-0.15) rectangle (-2,-0.35);
\fill [pattern color=nicered, pattern=north west lines] 
      (0.5,-0.15) rectangle (1,-0.35);
\draw [|-|] (-2,-0.5) -- (1,-0.5);
\draw (0,-1) node {$\vdots$};
\draw [|-|] (-1,-1.5) -- (2,-1.5);
\draw [|-|] (-0.5,-2) -- (2.5,-2);
\fill [pattern color=nicered, pattern=north west lines] 
      (-1,-1.65) rectangle (-0.5,-1.85);
\fill [pattern color=nicered, pattern=north west lines] 
      (2,-1.65) rectangle (2.5,-1.85);
\end{scope}
\end{tikzpicture}
\\[5ex]
\caption{From left to right: an "intersecting family" of intervals, a "chain", and an "anti-chain".}
\label{fig:chain-antichain}
\end{figure*}

\begin{lemma}\label{lem:chain-antichain}
Let $X$ be a partially ordered set and suppose that
all its "chains" and "anti-chains" have cardinality
at most $n$. Then the cardinality of $X$ is at most $n^2$.
\end{lemma}

\begin{proof}
To begin with, notice that $X$ is well-founded, due to the 
hypothesis that "chains" have cardinality at most $n$.
Define the partition $Y_1,Y_2, \ldots$ of $X$, where each
$Y_i$ contains all and only the \emph{minimal} elements of
$X \setminus \bigcup_{j<i} Y_j$ --- in particular, each $Y_i$ 
is defined inductively on the basis of the previous sets 
$Y_1,\dots,Y_{i-1}$.
By construction, every subset $Y_i$ is an "anti-chain",
and hence, by the hypotheses of the claim, it has cardinality
at most $n$. 

Let us now bound by $n$ the number of subsets of the partition.
Towards a contradiction, assume that $Y_1,Y_2,\dots,Y_{n+1}$ 
belong to the partition of $X$.
By construction, for every $1<i\le n+1$ and every $y\in Y_i$, 
there is $y'\in Y_{i-1}$ such that $y'<y$ 
(otherwise $y$ should have been added to $Y_{i-1}$).
Using this property and a simple induction, we can construct 
a "chain" of length $n+1$: we start by taking an arbitrary $y_{n+1}\in Y_{n+1}$
and then we repeatedly use the property to prepend to a "chain"
$y_i < y_{i+1} < \dots < y_{n+1}$, with $i>1$, 
$y_i\in Y_i$, $y_{i+1}\in Y_{i+1}$, \dots, $y_{n+1}\in Y_{n+1}$,
a new element $y_{i-1} < y_i$, with $y_{i-1}\in Y_{i-1}$.
Clearly, such a "chain" of length $n+1$ leads to a 
contradiction, and hence the partition $Y_1, Y_2, \ldots$ 
of $X$ contains at most $n$ elements.
We conclude that $|X| = \sum_{i}|Y_i|\le  n^2$.
\end{proof}

Ultimately, we aim at applying Lemma~\ref{lem:chain-antichain} 
to bound the cardinality of every "intersecting family" of 
"prefix-minimal" (resp., "suffix-minimal") intervals with 
pairwise distinct right (resp., left) endpoints, 
using the containment relation as partial order.
To this end, it is crucial to bound the cardinalities of 
the "chains" and "anti-chains" of such an "intersecting family".
It will be also convenient to avoid singleton intervals 
when reasoning about "intersecting families" 
(note that there is at most one singleton interval 
in every "intersecting family").

\begin{lemma}\label{lem:chain-bound}
Let $\cS$ be an "interval structure", 
$\varphi$ a $\BEpi$ formula, 
$\cI$ an "intersecting family" of non-singleton "prefix-minimal" 
(resp., "suffix-minimal") intervals for $\varphi$, 
with pairwise distinct right (resp., left) endpoints, 
and $\cI'$ a "chain" or an "anti-chain" of $\cI$, 
where the partial order is given by containment. 
We have that
\begin{equation}\label{eq:bound}
  |\cI'| ~\le~ 2^{2|\varphi|}.
\end{equation}
\end{lemma}

\begin{proof}
We present the proof for an "intersecting family" 
of non-singleton \emph{"prefix-minimal"} intervals for $\varphi$ 
(the case of "suffix-minimal" intervals uses symmetric arguments).
Towards a contradiction, assume that there exist 
a $\BEpi$ formula $\varphi$, 
an "intersecting family" 
$\cI$ of non-singleton "prefix-minimal" intervals for
$\varphi$ with pairwise distinct right endpoints, 
and a subset $\cI'$ of $\cI$ that is a "chain" or an 
"anti-chain" and that violates the bound \eqref{eq:bound}, 
i.e., $\cI'$ contains more than $2^{2|\varphi|}$ intervals.
We also assume, without loss of generality, that 
$\varphi$ is a smallest formula witnessing this 
violation of the bound
(later we will exploit this 
assumption when considering families of "prefix-minimal" 
intervals for subformulas of $\varphi$).

\AP
Let $""\BArgs*\varphi""$ (resp., $""\EArgs*\varphi""$) be
the set of formulas $\alpha$ such that $\hsB\alpha$ (resp., $\hsE\alpha$) is
a subformula of $\varphi$ with no other modal operator above it.
For example, if $\varphi = \hsB\alpha_1 \:\wedge\: \hsB\hsB\alpha_2 \:\wedge\: \hsE\hsB\alpha_3$,
then $\BArgs\varphi = \{\alpha_1,\hsB\alpha_2\}$ and $\EArgs = \{\hsB\alpha_3\}$.
Note that $|\varphi| \geq |\BArgs\varphi| + |\EArgs\varphi|$.

\AP
Define the ""$\varphi$-profile"" of a non-singleton interval 
$I$ as the pair $(B,E)$, where $B$ (resp., $E$) is the set 
of formulas $\alpha\in\BArgs{\varphi}$
(resp., $\alpha\in\EArgs{\varphi}$) that "hold at" 
"prefixes" (resp., "suffixes") of $I$. 
Note that any two non-singleton intervals 
with the same "$\varphi$-profile" 
either both "satisfy@holds at" $\varphi$ 
or both "satisfy@holds at" $\neg\varphi$;
in particular, this holds thanks to the 
fact that $\varphi$ is in "homogeneous normal form".

We also observe that 
there are at most $2^{|\BArgs{\varphi}| + |\EArgs{\varphi}|}$ 
distinct "$\varphi$-profiles". 
Therefore, by our assumption on $\cI'$, there are
\[
  n ~>~ 2^{2|\varphi| - |\BArgs{\varphi}| - |\EArgs{\varphi}|}
\]
intervals $I_1,\dots,I_n \in \cI'$ with the same $\varphi$-profile.
Without loss of generality, assume that the intervals $I_1,\dots,I_n$ 
are listed based on the natural ordering of their right endpoints, 
that is, $\max(I_1) < \dots < \max(I_n)$.
Depending on $\cI'$ being a "chain" or an "anti-chain", 
the left endpoints of these intervals are also ordered, 
in descending, resp., ascending order
(see Figure \ref{fig:chain-antichain}).

For the rest of the proof, unless otherwise stated, 
$i$ will denote a natural number from $1$ to $n-1$, and will 
be used in particular to index pairs of consecutive 
intervals, say $I_i$ and $I_{i+1}$.
\AP
For every $i$, let
$w_i < x_i \leq y_i < z_i$ be the four endpoints of $I_i$ and $I_{i+1}$.
Further, let $""\leftDeltaIntervals*{i}"" = [w_i+1,x_i]$ and
$""\rightDeltaIntervals*{i}"" = [y_i,z_i-1]$ (these intervals are represented 
by the red dashed rectangles in Figure \ref{fig:chain-antichain}).
Thanks to the fact that $\cI'$ is a "chain" or an "anti-chain", 
the $\leftDeltaIntervals{i}$'s and the $\rightDeltaIntervals{i}$'s 
are pairwise disjoint across all $i$ 
(this property will be used later and is the main reason 
for restricting our attention to "chains" and "anti-chains").

\AP
Given $\alpha\in\BArgs{\varphi}$ and $1\le i<n$, 
a ""special $\hsB*$-witness of $\alpha$ at $i$"" (if it exists)
is the "prefix-minimal" interval for $\alpha$ that has the 
same left endpoint as $I_{i+1}$ and whose right endpoint belongs to $\rightDeltaIntervals{i}$.
Figure \ref{fig:special-witnesses} gives two examples
of "special $\hsB*$-witnesses", represented by green brackets:
one example is for the "chain" arrangement and the other
is for the "anti-chain" arrangement.
\AP
Symmetrically, given $\alpha\in\EArgs{\varphi}$ and $1\le i<n$,
a ""special $\hsE*$-witness of $\alpha$ at $i$"" (if it exists) is the "suffix-minimal"
interval for $\alpha$ that has the same right endpoint as $I_i$ and whose 
 left endpoint belongs to $\leftDeltaIntervals{i}$.
"Special $\hsE*$-witnesses" are represented in Figure
\ref{fig:special-witnesses} by blue brackets.

Now, we tag an index $1\le i<n$ with a pair $(B,\alpha)$
(resp., $(E,\alpha)$) whenever
$\alpha\in\BArgs{\varphi}$ (resp., $\alpha\in\EArgs{\varphi}$) and
there is a "special $\hsB*$-witness" (resp., "$\hsE*$-witness@special $\hsE*$-witness") 
of $\alpha$ at $i$.
If there is no "special witness" for any $\alpha$, 
then we tag $i$ with the symbol $\bot$.
Let $K_i = [\min(I_{i+1}),\max(I_i)]$ and observe
that $K_i$ is a \emph{proper "prefix"} of $I_{i+1}$.
We will prove that, for some index $i$, the interval 
$K_i$ "satisfies@holds at" $\varphi$, thus contradicting 
"prefix-minimality" of $I_{i+1}$.
Towards this, it will be sufficient to find an 
index $i$ tagged with $\bot$. 
Indeed, if this happens, then we claim that

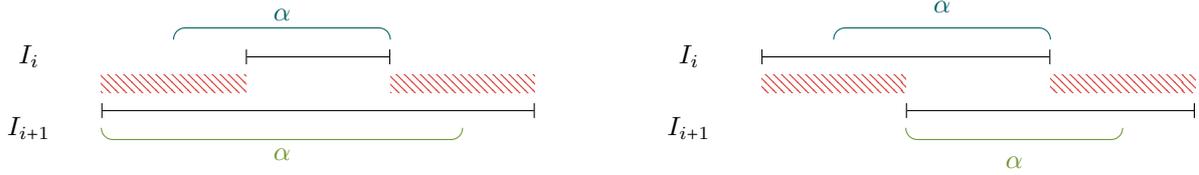
\begin{figure*}[t!]
\centering
\begin{tikzpicture}[scale=0.95]
\begin{scope}
\coordinate (A) at (1,0);
\coordinate (B) at (2,0);
\coordinate (C) at (4,0);
\coordinate (D) at (0,-0.75);
\coordinate (E) at (5,-0.75);
\coordinate (F) at (6,-0.75);
\draw [|-|] (B) -- (C);
\draw [|-|] (D) -- (F);
\draw [nicecyan, above round brace] ([yshift=7] A) -- ([yshift=7] C)
      node [above=-0.1, midway, yshift=8] {$\alpha$};
\draw [nicegreen, below round brace] ([yshift=-7] D) -- ([yshift=-7] E)
      node [below=-0.1, midway, yshift=-8] {$\alpha$};
\fill [pattern color=nicered, pattern=north west lines] 
      ([yshift=-7] B) rectangle ([yshift=7] D);
\fill [pattern color=nicered, pattern=north west lines] 
      ([yshift=-7] C) rectangle ([yshift=7] F);
\draw (-1,0) node {$I_i$};
\draw (-1,-1) node {$I_{i+1}$};
\end{scope}
\begin{scope}[xshift=260]
\coordinate (A) at (0,0);
\coordinate (B) at (1,0);
\coordinate (C) at (4,0);
\coordinate (D) at (2,-0.75);
\coordinate (E) at (5,-0.75);
\coordinate (F) at (6,-0.75);
\draw [|-|] (A) -- (C);
\draw [|-|] (D) -- (F);
\draw [nicecyan, above round brace] ([yshift=7] B) -- ([yshift=7] C)
      node [above, midway, yshift=8] {$\alpha$};
\draw [nicegreen, below round brace] ([yshift=-7] D) -- ([yshift=-7] E)
      node [below, midway, yshift=-8] {$\alpha$};
\fill [pattern color=nicered, pattern=north west lines] 
      ([yshift=-7] A) rectangle ([yshift=7] D);
\fill [pattern color=nicered, pattern=north west lines] 
      ([yshift=-7] C) rectangle ([yshift=7] F);
\draw (-1,0) node {$I_i$};
\draw (-1,-1) node {$I_{i+1}$};
\end{scope}
\end{tikzpicture}
\caption{"Special witnesses" in a "chain" (left) 
and in an "anti-chain" (right).}
\label{fig:special-witnesses}
\end{figure*}

\begin{claim}
Every formula 
$\alpha\in\BArgs{\varphi}$ (resp., $\alpha\in\EArgs{\varphi}$) 
that "holds at" a "prefix" (resp., "suffix") of $I_i$ also 
"holds at" a "prefix" (resp., "suffix") of $K_{i}$, and 
vice versa. 
\end{claim}

The above claim would then imply that the 
"$\varphi$-profile" of $K_i$ coincides with 
that of $I_i$, and hence $K_i\models\varphi$.

\begin{proof}[Proof of the claim]
Assume that index $i$ is tagged with $\bot$.
Consider some $\alpha\in\BArgs{\varphi}$.
If $\alpha$ "holds at" a "prefix" of $I_i$, 
then $\alpha$ "holds at" some "prefix"
of $I_{i+1}$ as well, because $I_i$ and $I_{i+1}$ 
have the same "$\varphi$-profile".
Let $J$ be the smallest "prefix" of $I_{i+1}$ 
that "satisfies@holds at" $\alpha$.
Due to $i$ being tagged with $\bot$, we have that 
$\max(J) < \max(I_i) = \max(K_i)$, meaning that 
$J$ is also a "prefix" of $K_i$.
Conversely, if $\alpha$ "holds at" a "prefix" of $K_i$, 
then it trivially "holds at" a "prefix" of $I_{i+1}$ as well, 
and thus it "holds at" a "prefix" of $I_{i}$, too,
because $I_i$ and $I_{i+1}$ have the same "$\varphi$-profile".
Next, consider some $\alpha\in\EArgs{\varphi}$.
If $\alpha$ "holds at" some "suffix" of $I_i$, then let $J$ be the smallest "suffix" of $I_{i}$ that "satisfies@holds at" $\alpha$.
Due to $i$ being tagged with $\bot$, 
we have that $\min(J) > \min(I_{i+1}) = \min(K_i)$, 
meaning that $J$ is also a "suffix" of $K_i$.
Conversely, assume that $\alpha$ "holds at" some 
"suffix" of $K_i$ and let $J$ be the smallest "suffix" 
of $K_{i}$ that "satisfies@holds at" $\alpha$.
Once again, since $i$ is tagged with $\bot$, 
we have that $\min(J) > \min(I_{i})$, meaning 
that $J$ is a "suffix" of $I_i$, too.
\end{proof}
%

It remains to prove that at least one index $i$ is tagged with $\bot$.
For this, we bound the number of indices tagged with pairs of the form 
$(X,\alpha)$, with $X=B$ (resp., $X=E$) and $\alpha\in\XArgs{\varphi}$.
By construction, for each tag $(X,\alpha)$, the 
"special $\hsX*$-witnesses" of $\alpha$ form an 
"intersecting" ("anti@anti-chain"-)"chain" $\cI_{X,\alpha}$ of
"prefix-minimal" (resp., "suffix-minimal") 
intervals for $\alpha$. 
Moreover, we know that:
\begin{itemize}
\item All intervals in $\cI_{X,\alpha}$ are non-singleton.

      This is because the only scenario where a 
      \emph{singleton} "special $\hsX*$-witness" arises is when 
      $\cI'$ is "anti-chain", $n=2$, and $\max(I_1) = \min(I_2)$.
      This scenario is however excluded by the fact that
      $n > 2^{2|\varphi| - |\BArgs{\varphi}| - |\EArgs{\varphi}|} \ge 2$.
\item The intervals in $\cI_{X,\alpha}$ have pairwise distinct
      right (resp., left) endpoints.
      
	  This is because those endpoints belong to
	  the intervals $\rightDeltaIntervals{i}$
	  (resp., $\leftDeltaIntervals{i}$), which
	  are pairwise disjoint across all $i$'s.
\item The cardinality of each ("anti@anti-chain"-)"chain" 
      $\cI_{X,\alpha}$ is at most $2^{2|\alpha|}$.

      This is thanks to the previous properties and because 
      $\alpha$ is a proper subformula of $\varphi$, which 
      was assumed to be a smallest formula violating 
      the bound \eqref{eq:bound}.
\end{itemize}
In view of the last property, we derive that the number
of indices that are \emph{not} tagged with $\bot$ is 
\begin{align*}
  n' &~\le~ \sum\nolimits_{\alpha\in\BArgs{\varphi}} 2^{2|\alpha|} 
            ~+~ 
            \sum\nolimits_{\alpha\in\EArgs{\varphi}} 2^{2|\alpha|}
  \\
     &~\le~ 2^{\sum_{\alpha\in\BArgs{\varphi}} 2|\alpha| \:+\: 
               \sum_{\alpha\in\EArgs{\varphi}} 2|\alpha| } \ ,
\end{align*}
where the last inequality follows from  majorating sums with products.
Next, recall that
$n > 2^{2|\varphi| - |\BArgs{\varphi}| - |\EArgs{\varphi}|}$,
and hence the number of indices $1\le i<n$ that are tagged 
with $\bot$ is
\[
  n - 1 - n' 
  ~\ge~ 2^{2|\varphi| - |\BArgs{\varphi}| - |\EArgs{\varphi}|}
        ~-~
        2^{\sum_{\alpha\in\BArgs{\varphi}} 2|\alpha| \:+\: 
           \sum_{\alpha\in\EArgs{\varphi}} 2|\alpha|} \ .
\]
We prove that the right hand-side number is always positive 
by showing that 
$2|\varphi| - |\BArgs{\varphi}| - |\EArgs{\varphi}| 
 > 
 \sum_{\alpha\in\BArgs{\varphi}} 2|\alpha| \:+\: 
 \sum_{\alpha\in\EArgs{\varphi}} 2|\alpha|$.
We distinguish two cases, depending on whether
or not $\varphi$ contains modal operators.
If $\varphi$ contains no modal operators, then
$2|\varphi| - |\BArgs{\varphi}| - |\EArgs{\varphi}| 
 = 2|\varphi| 
 > 0
 = \sum_{\alpha\in\BArgs{\varphi}} 2|\alpha| \:+\: 
   \sum_{\alpha\in\EArgs{\varphi}} 2|\alpha|$.
Otherwise, if $\varphi$ contains at least one modal 
operator, then we observe that \emph{(i)} the size of $\varphi$
is at least the sum of the sizes of the subformulas 
$\hsX\alpha$, for $X \in \{ B, E \}$ and $\alpha\in\XArgs{\varphi}$,
which are $|\hsX\alpha| = |\alpha| + 1$,
and \emph{(ii)}
$\sum_{\alpha\in\XArgs{\varphi}}(|\alpha| + 1) = 
 \big(\sum_{\alpha\in\XArgs{\varphi}}|\alpha|\big) + |\XArgs{\varphi}|$, for $X \in \{ B, E \}$.
From this we derive:
\begin{align*}
  & 2|\varphi| - |\BArgs{\varphi}| - |\EArgs{\varphi}|
  \\[1ex]
  \ge~ &\sum\nolimits_{\alpha\in\BArgs{\varphi}} 2|\alpha| 
        ~+~
        \sum\nolimits_{\alpha\in\EArgs{\varphi}} 2|\alpha|
        ~+~
        |\BArgs{\varphi}| ~+~ |\EArgs{\varphi}| 
  \\[1ex]
  >~  &\sum\nolimits_{\alpha\in\BArgs{\varphi}} 2|\alpha| 
        ~+~
        \sum\nolimits_{\alpha\in\EArgs{\varphi}} 2|\alpha| \ .
\end{align*}
We have just shown that at least one index $i$ 
must be tagged with $\bot$, which completes the proof
of the lemma.
\end{proof}

Putting together Lemmas \ref{lem:chain-antichain}
and \ref{lem:chain-bound}, we obtain the desired bound
for an arbitrary "intersecting family" of non-singleton
"prefix@prefix-minimal"/"suffix-minimal" intervals for 
$\varphi$:


\begin{restatable}{corollary}{CoreBoundCorollaryI}
\label{cor:core-bound-i}
Let $\cS$ be an "interval structure", 
$\varphi$ a $\BEpi$ formula, 
and $\cI$ an "intersecting family" of non-singleton 
"prefix-minimal" (resp., "suffix-minimal") intervals 
for $\varphi$. 
Then the number of distinct right (resp., left) endpoints
of intervals of $\cI$ is at most $2^{4|\varphi|}$.
\end{restatable}

We conclude this part by showing how "prefix-minimal" intervals
for $\varphi$ can be characterized using fresh 
propositional letters
and suitable formulas $\dec{\varphi}$ and $\enc{\varphi}$ 
(a similar corollary can be stated for "suffix-minimal" intervals).

\begin{restatable}{corollary}{CoreBoundCorollaryII}
\label{cor:core-bound-ii}
Consider a $\BEpi$ formula $\varphi$ 
over a "signature" $\signature$
and let $\signature' = \signature \uplus\{p_1,\dots,p_m,\ell,r,s\}$,
where $p_1,\dots,p_m$, $\ell,r,s$ are fresh propositional letters
and $m=4|\varphi|$ (this $m$ is precisely the exponent appearing 
in the bound of Corollary \ref{cor:core-bound-i}).
\AP
Define the $\BEpi$ formulas%
\footnote{Note that, despite the notation, 
          the formula $\dec{\varphi}$ only depends 
          on the signature and the size of $\varphi$,
          whereas $\enc{\varphi}$ depends entirely on $\varphi$.}
\begin{align*}
  ""\dec*{\varphi}"" 
  &~=~~ \Big( \hsB(\pi\wedge\ell) ~\wedge~ \hsE(\pi\wedge r) 
        ~\wedge~ \\[-1ex]
  &~~~~ \phantom{\Big(} ~~ \bigwedge\nolimits_{i=1,\dots,m} \big(\hsB(\pi\wedge p_i) \iff 
                                                \hsE(\pi\wedge p_i)\big) \Big) \\
  &~\vee~~ (\pi \wedge s) 
  \\[2ex]
  ""\enc*{\varphi}""
  &~=~~ 
  \underbrace{\big(\dec{\varphi} \wedge \neg\hsB\dec{\varphi} ~\rightarrow~ \varphi \big)}%
             _{\midward{\text{\scriptsize\qquad\qquad\quad
                                         prefix-minimial intervals for $\dec{\varphi}$ satisfy $\varphi$}}}  \\[1ex]
  &~\wedge~~ 
  \underbrace{\big(\varphi \wedge \neg\hsB\dec{\varphi} ~\rightarrow~ \dec{\varphi} \big)}%
             _{\midward{\text{\scriptsize\qquad\qquad\quad 
                                         prefix-minimal intervals for $\varphi$ satisfy $\dec{\varphi}$}}}
\end{align*}
We have that $\enc{\varphi}$ is an "expander" from 
$\signature$ to $\signature'$
and that $\hsB\varphi$ and $\hsB\dec{\varphi}$ are 
"equivalent" over all "interval structures" that make 
$\enc{\varphi}$ "valid@valid over".
\end{restatable}

\begin{proof}
Let us first explain the intended use of the fresh propositional
letters $p_1,\dots,p_m$, $\ell,r,s$.
The letters $p_1,\dots,p_m$ will annotate points of an 
"interval structure" with $m$-tuples of bits, thus enumerating
an exponentially-large set (e.g.,~$\{1,\dots,2^m\}$). 
More precisely, the left and right endpoints of every non-singleton 
"prefix-minimal" interval for $\varphi$ will be identified by having 
labels $\ell$ and $r$, respectively, and the same $m$-tuple of bits 
--- this correlation between endpoints is checked by the first disjunct 
of $\dec{\varphi}$.
Singleton "prefix-minimal" intervals for $\varphi$ will instead be 
identified by the special label $s$ --- this is checked in the second
disjunct of $\dec{\varphi}$.
%
%
Another important constraint is that every two intersecting intervals 
that are non-singleton, "prefix-minimal" for $\gamma$, and not a 
"suffix" one of another will have their endpoints marked by different 
$m$-tuples of bits.

Corollary \ref{cor:core-bound-i} guarantees the existence 
of an annotation satisfying all the above constraints.
Such an annotation is enforced precisely by the formula $\enc{\varphi}$,
which turns out to be an "expander" from $\signature$ to $\signature'$
(namely, every "interval structure" over $\signature$ admits an 
"expansion" over $\signature'$ that makes $\enc{\varphi}$ "valid@valid over").
Moreover, if the annotation is correct, namely, if 
$\enc{\varphi}$ is "valid over" an expanded "interval structure",
then every "prefix-minimal" interval for $\varphi$ is also 
a "prefix-minimal" interval for $\dec{\varphi}$, and vice versa.
Note that there may still exist intervals that "satisfy" 
$\varphi$ but not $\dec{\varphi}$, or vice versa; however,
those intervals will always contain proper "prefixes" that
"satisfy" both $\varphi$ and $\dec{\varphi}$.
Overall, this proves that the two formulas $\hsB\varphi$ 
and $\hsB\dec{\varphi}$ are "equivalent" over
expanded "interval structures" that make $\enc{\varphi}$
"valid@valid over".
\end{proof}

\subsection{Normalization procedure}\label{subsec:proof-of-theorem}

We are now ready to describe the "normalization" procedure underlying 
Theorem \ref{thm:main}. Let $\psi$ be a $\BEpi$ formula.
The "normalization" of $\psi$ consists of repeatedly applying some 
rewriting steps that preserve "satisfiability" and progressively 
reduce the number of distinct subformulas of "depth" larger than $2$, 
until a "shallow normal form" is eventually obtained.

Every rewriting step is applied to a formula of the form 
$\psi_i \:\wedge\: \hsGu\xi_i$ over a "signature" $\signature_i$ 
(initially, $\psi_0 = \psi$, $\xi_i = \true$, and $\signature_i = \signature$), 
and results in an "equi-satisfiable" formula $\psi_{i+1} \:\wedge\: \hsGu\xi_{i+1}$ 
over an extended "signature" $\signature_{i+1}$.
To perform the rewriting step, we must choose a subformula $\hsX\varphi$ 
of $\psi_i \:\wedge\: \hsGu\xi_i$, for some $X\in\{B,E\}$, 
that has "depth" $d>2$ and that does not occur under the scope of any other modal 
operator, except possibly the operator $\hsGu$ that has $\xi_i$ as argument.
We then use Corollary \ref{cor:core-bound-ii} to obtain an "expander" 
$\enc{\varphi}$ from $\signature_i$ to $\signature_{i+1}$ and a 
formula $\hsX\dec{\varphi}$ "equivalent" to $\hsX\varphi$ over every 
"interval structure" that makes $\enc{\varphi}$ "valid@valid over".
We then rewrite $\psi_i \:\wedge\: \hsGu\xi_i$ into the formula
\begin{align*}
\begin{array}{c}
  \underbrace{\psi_i\big[\hsX\varphi \:/ \hsX\dec{\varphi}\big]}_{\psi_{i+1}}
    ~\wedge~ \\
  \hsGu\big( \underbrace{\xi_i\big[\hsX\varphi \:/ \hsX\dec{\varphi}\big] 
                           \:\wedge\: 
                           \enc{\varphi}}_{\xi_{i+1}} \big) 
\end{array}
  \tag{$\dagger$}
\end{align*}
Thanks to distributivity of $\hsGu$ with respect to $\wedge$, 
the formula ($\dagger$) is "equivalent" to 
$(\psi_i \:\wedge\: \hsGu\xi_i)[\hsX\varphi \:/ \hsX\dec{\varphi}]
  ~\wedge~ 
  \hsGu\enc{\varphi}$.
Moreover, thanks to Lemma \ref{lem:expander-equivalence}, 
the latter formula is "equi-satisfiable" as $\psi_i \:\wedge\: \hsGu\xi_i$.
This completes the description of a rewriting step.

Let us now analyse the complexity of the "normalization" procedure.
The procedure terminates when one cannot choose any subformula 
$\hsX\varphi$ with the desired properties: in this case the rewritten formula
$\psi_i \:\wedge\: \hsGu\xi_i$ turns out to be in "shallow normal form"
and we can let $\psi^\star = \psi_i \:\wedge\: \hsGu\xi_i$.
To bound the number of rewriting steps, we study how a single 
rewriting step affects the number of distinct subformulas
of "depth" larger than $2$.
As for $\hsX\varphi$, we observe that this subformula does not 
occur anymore in the rewritten formula $\psi_{i+1} \:\wedge\: \hsGu\xi_{i+1}$
(in particular, the inteded use of $\enc{\varphi}$ is to entail
$\hsX\varphi \:\leftrightarrow\: \hsX\dec{\varphi}$, but the chosen 
writing in the statement of Corollary \ref{cor:core-bound-ii}
avoids having $\varphi$ under the scope of a modal operator, 
thus guaranteeing that $\enc{\varphi}$ has "depth" at most $2$). 
On the other hand, new occurrences of subformulas may emerge in 
$\psi_{i+1} \:\wedge\: \hsGu\xi_{i+1}$:
these are either formulas of "depth" at most $2$ 
(e.g., $\hsX\dec{\varphi}$) or copies of formulas that already occur
in $\psi_i \:\wedge\: \hsGu\xi_i$ (e.g., $\varphi$).
Summing up, the effect of a rewriting step is to decrease the 
number of \emph{distinct} subformulas of "depth" larger than $2$.
This implies that the number of rewriting steps is at most linear in
the size of the original formula $\psi$.
Finally, each rewriting step is purely syntactical and can be carried out 
efficiently on the involved formula $\psi_i \:\wedge\: \hsGu\xi_i$,
whose size grows at most linearly with $i$.
This shows that the entire "normalization" procedure 
can be performed in polynomial time w.r.t.~$|\psi|$,
and completes the proof of Theorem \ref{thm:main}.
\qed


\section{Complexity of the satisfiability problem}\label{sec:complexity}

In this section, we build up on the previous normalization result
to prove a tight complexity bound:

\begin{theorem}\label{thm:complexity}
The "satisfiability" problem for $\BE$ logic restricted to
"homogeneous" "interval structures" is $\expspace$-complete.
\end{theorem}

An $\expspace$ lowerbound for $\BE$ under "homogeneity"
was already proven in \cite{DBLP:journals/tcs/BozzelliMMPS19}, 
so we focus on the upperbound.
In view of Proposition \ref{prop:homogeneus-normal-form} and
Theorem \ref{thm:main}, given any $\BE$ formula $\psi$, one
can compute in polynomial time a $\BEpi$ formula $\psi^\star$
that is "equi-satisfiable" over "homogeneous" "interval structures".
Of course, this also means that $\psi^\star$ has size at most 
polynomial in $|\psi|$.
We argue below that one can test "satisfiability" of 
a formula in "shallow normal form"
in exponential space with respect to the size of the formula itself. 
Together with the previous observations, this proves Theorem \ref{thm:complexity}.



\knowledge{notion}
	| logical type
	| types
	| type
	| depth-$0$ type
	| depth-$0$ types
	| depth-$d$ type
	| depth-$d$ types
\knowledge{notion}
	| depth-$1$ type
	| depth-$1$ types
\knowledge{notion}
	| depth-$2$ $\varphi$-type
	| depth-$2$ $\varphi$-types
	| depth-$2$ $\psi$-type
	| depth-$2$ $\psi$-types
	| depth-$2$ type
	| depth-$2$ types
\knowledge{notion}
	| \lenone
	| \lenone*
\knowledge{notion}
	| \lentwo
	| \lentwo*
\knowledge{notion}
	| \lenthree
	| \lenthree*
\knowledge{notion}
	| depth-$0$ composition
	| depth-$d$ composition
\knowledge{notion}
	| depth-$1$ composition
\knowledge{notion}
	| depth-$2$ composition
\knowledge{notion}
	| \Depth
	| \Depth*{1}{\varphi}
\knowledge{notion}
	| \vdash
	| \cT \vdash* \alpha
\knowledge{notion}
	| dummy depth-$1$ type $\dummy*$
	| dummy depth-$1$ type
	| dummy type
	| dummy 
\knowledge{notion}
	| left context
	| left contexts
	| context
	| contexts
\knowledge{notion}
	| right context
	| right contexts

\subsection{Composition of logical types}\label{subsec:types}

We need to formalize a notion of 
"logical type", similar to the notion of "profile" used in the proof of 
Lemma \ref{lem:chain-bound}, that not only determines which formulas
"hold at" a given interval, but also satisfies mild compositional 
properties, that is, under suitable conditions, one can compute the 
"type" of the sum of two adjacent intervals on the basis of the "types" of the 
original intervals.
It will be convenient to define "types" separately for formulas of 
"depth" $0$, $1$, and $2$ (there is no need to consider higher "depths", 
as we assume to deal with formulas in "shallow normal form").
We will first present the rather simple definitions and properties of 
"depth-$0$@depth-$0$ types" and "depth-$1$ types", and then focus on the 
more complex notion of "depth-$2$ type".

\paragraph{Depth-$0$ and depth-$1$ types}
We fix, once and for all, an "interval structure" $\cS=\intervalstructure$
and we assume that all formulas are over the "signature" $\signature$ of $\cS$.

\begin{definition}\label{def:depth0-type}
\emph{The ""depth-$0$ type""} of an interval $I$,
denoted $\type{0}{}(I)$, is either the set 
$\{\pi\} \cup \{p\in\signature \::\: \cS,I\models \pi \wedge p\}$
or the empty set, depending on whether $I$ is a singleton or not.

\AP
\emph{The ""depth-$1$ type""} of an interval $I=[x,y]$
is the quadruple $\type{1}{}(I)=(S,T,B,E)$, where
$S$ is the symbol $""\lenone*""$, $""\lentwo*""$, or
$""\lenthree*""$, depending on whether $I$ contains one point,
two points, or more, $T = \type{0}{}(I)$, $B = \type{0}{}([x,x])$, and
$E = \type{0}{}([y,y])$.
\end{definition}

It is easy to see that "depth-$0$@depth-$0$ types" 
(resp., "depth-$1$@depth-$1$ types") "types" of 
adjacent intervals can be composed to form 
the "depth-$0$@depth-$0$ type" 
(resp., "depth-$1$@depth-$1$ type") 
"type" of the sum of the two intervals.
One can also verify that the "depth-$0$@depth-$0$ type" 
(resp., "depth-$1$@depth-$1$ type") 
"type" of an interval determines which
formulas of "depth" $0$ (resp., "depth" at most $1$) 
"hold at" that interval. 
These simple results are formalized in the next two lemmas 
below.

\begin{restatable}{lemma}{DepthOneComposition}
\label{lem:depth1-composition}
\AP
For both $d=0$ and $d=1$, there is a 
""composition@depth-$0$ composition""%
""@depth-$1$ composition""
operator $\cdot{d}$ on "depth-$d$ types" 
that is computable in polynomial time and such that,
for all pairs of adjacent intervals $I,J$, with $\max(I)+1=\min(J)$,
$\type{d}{}(I) \cdot{d} \type{d}{}(J) = \type{d}{}(I\cup J)$.
\end{restatable}

\begin{proof}
The "composition@depth-$0$ composition" 
of "depth-$0$ types" is trivial:
for every pair of "depth-$0$ types"
$T,T'$, we simply let $T\cdot{0} T' = \emptyset$.
This is correct because the sum of two adjacent
intervals always results in a non-singleton interval,
whose "depth-$0$ type" is the empty set.

As for the "composition@depth-$1$ composition" 
of two "depth-$1$ types", say
$\cT=(S,T,B,E)$ and $\cT'=(S',T',B',E')$, we let
$\cT \cdot{1} \cT' = (S'', T\cdot{0} T', B, E')$,
where $S''$ is either $\lentwo$ or $\lenthree$ 
depending on whether $S = S' = \lenone$ or not, 
and $T \cdot{0} T'$ is the composition of the
"depth-$0$ types" $T$ and $T'$, as defined just above.
It is immediate to check that if 
$\cT=\type{1}{}([x,y])$
and $\cT'=\type{1}{}([y+1,z])$, then 
$\type{1}{}([x,z]) = \cT \cdot{1} \cT'$.
\end{proof}

\begin{restatable}{lemma}{DepthOneEvaluation}
\label{lem:depth1-evaluation}
For both $d=0$ and $d=1$, 
for every $\BEpi$ formula $\varphi$ of "depth" at most $d$, 
and for all intervals $I,J$ such that 
$\type{d}{}(I)=\type{d}{}(J)$, we have
$\cS,I\models\varphi$ iff $\cS,J\models\varphi$.
Moreover, whether $\cS,I\models\varphi$ holds or not can be decided in
polynomial time given $\varphi$ and $\type{d}{}(I)$.
\end{restatable}

\begin{proof}
We first prove the claim for $d=0$.
For the case $\varphi = \pi$, we have $\cS,I\models\varphi$ 
if and only if $I$ is a singleton, or, equally, $\pi \in \type{0}{}(I)$.
The case $\varphi = \pi \wedge p$ is trivial as well, as we have
$\cS,I\models\varphi$ if and only if $p\in\type{0}{}(I)$.
It remains to consider the case where $\varphi$ is a 
Boolean combination of the previous atomic formulas. 
In this case, we determine the "evaluation" of $\varphi$ at $I$ 
``homomorphically''  on the basis of the "evaluations" of the atomic formulas. 

Let us now prove the claim for $d=1$. The interesting cases are when 
$\varphi$ has "depth" $0$ or it is of the form $\hsB\alpha$ or $\hsE\alpha$, 
with $\alpha$ again of "depth" $0$. Once the claim is proved for these cases, 
it can be generalized to Boolean combinations of those formulas using the 
same arguments as before.
Let $I=[x,y]$ and $\type{1}{}(I)=(S,T,B,E)$, and recall that
$T=\type{0}{}(I)$, $B=\type{0}{}([x,x])$, and $E=\type{0}{}([y,y])$.

If $\varphi$ has "depth" $0$, then we know that the component
$T$ ($=\type{0}{}(I)$) already determines whether or not $\cS,I\models\varphi$.

If $\varphi = \hsB\alpha$, we further distinguish three subcases, 
depending on $S$.
If $S=\lenone$, then $I$ is a singleton and hence
$\cS,I\not\models\hsB\alpha$.
If $S=\lentwo$, then the only prefix of $I$ is the singleton interval $[x,x]$,
hence $\cS,I\models\hsB\alpha$ iff $\cS,[x,x]\models\alpha$.
Since $\alpha$ has "depth" $0$, the latter condition can be 
decided using the "type@depth-$0$ type" $B=\type{0}{}([x,x])$.
If $S=\lenthree$, then since $\alpha$ is a Boolean combination 
of formulas of the form $\pi$ or $\pi \wedge p$, with $p\in\signature$,
it suffices to consider only two prefixes of $I$:
the singleton interval $J_0=[x,x]$ and the interval $J_1=[x,x+1]$.
In particular, we have $\cS,I\models\hsB\alpha$ if and only if $\cS,J_0\models\alpha$
or $\cS,J_1\models\alpha$. 
Again, the latter two conditions are determined by the "depth-$0$ types"
of $J_0$ and $J_1$, which are $B$ and $\emptyset$, respectively. 
This shows how to determine whether $\cS,I\models\hsB\alpha$
using the "type@depth-$1$ type" $\type{1}{}(I)=(S,T,B,E)$.

The remaining case is that of a formula $\varphi = \hsE\alpha$,
which can be handled by symmetric arguments, using the component 
$E$ instead of $B$.
\end{proof}

\paragraph{Depth-$2$ types}
We now introduce "types" for "depth"-$2$ formulas.
The machinery here is not as neat as one could hope, 
as there is a trade-off between the desired compositional 
properties and the number of possible "depth-$2$ types".
As an example, full compositionality of "types" for "depth"-$2$ formulas 
can only hold if we allow doubly exponentially many "types"  with respect to the 
size of the underlying "signature" --- this can be shown formally using 
arguments based on communication complexity and the fact that a 
"depth"-$2$ formula can describe a Stockmeyer's counter of 
level $2$ \cite{stockmeyer1974}.
In order to ease compositional properties while maintaining the 
number of "types" as low as possible, we will parameterise 
"depth-$2$ types" by a formula and some "contexts".

We first discuss a couple of tentative definitions, with their drawbacks.
Following the same principle used to define "depth-$1$ types", 
one may define the depth-$2$ type of an interval $I$ as $(\cT,\sB,\sE)$,
where $\cT=\type{1}{}(I)$, $\sB=\{ \type{1}{}(J) \::\: J\prefix I\}$, 
and $\sE=\{ \type{1}{}(J) \::\: J\suffix I\}$.
This notion of depth-$2$ type would be fully compositional 
and would determine the evaluation of every "depth"-$2$ formula
in "homogeneous normal form" (proofs omitted).
Unfortunately, there could be doubly exponentially many such types 
with respect to the size of the "signature", and this would not be 
compatible with the intended use that we will make 
in the satisfiability procedure.
Another option would be to parameterise the depth-$2$ type of $I$ 
by a formula $\varphi$ and define it as the triple $(\cT,\sB,\sE)$, 
where $\cT=\type{1}{}(I)$ as before, and $\sB$ (resp., $\sE$) is the 
set of subformulas $\alpha$ of $\varphi$ that "hold at" "proper prefixes" 
(resp., "suffixes") of $I$.
Of course, the resulting type would determine the evaluation of
$\varphi$ at the interval $I$.
This second attempt would also generate at most exponentially 
many "depth"-$2$ type with respect to the size of $\varphi$.
On the other hand, the resulting types would not carry enough information 
to be composable, the reason being that it is not sufficient to know 
which "depth"-$1$ subformulas "hold at" two adjacent intervals in order to 
derive which "depth"-$1$ subformulas "hold at" the union interval.
The appropriate notion of "depth"-$2$ type is somehow a blend of 
the two attempts that we have just discussed. 

Let us now fix some other useful notation and terminology:
\begin{itemize}
\item 
\AP
Given a $\BEpi$ formula $\varphi$, we denote by $""\Depth*{1}{\varphi}""$ 
the set of subformulas of $\varphi$ of "depth" at most $1$.
\item 
Lemma \ref{lem:depth1-evaluation} states that 
the "depth-$1$ type" of an interval $I$ effectively 
determines which formulas of "depth" at most $1$ "hold at" $I$.
This motivates the following notation: 
given a "depth-$1$ type" $\cT$ and a 
formula $\alpha\in\Depth{1}{\varphi}$, 
\AP
we write $""\cT \vdash* \alpha""$ to state 
that $\cS,I \models \alpha$ for some (or, equally, for every) 
interval $I$ such that $\type{1}{}(I)=\cT$
(this latter property can be tested efficiently given $\cT$ and $\alpha$).
\item 
By Lemma \ref{lem:depth1-composition}, "depth-$1$ types" 
are equipped with a "composition@depth-$1$ composition" 
operation $\cdot{1}$ that forms a semigroup structure. 
\AP
We complete the structure into a monoid by introducing 
the ""dummy depth-$1$ type $\dummy*$"" and by assuming 
that $\dummy \cdot{1} \cT = \cT \cdot{1} \dummy = \cT$ 
for every "depth-$1$ type" $\cT$.
\end{itemize}

\begin{definition}\label{def:depth2-type}
Let $\cL,\cR$ be some (possibly "dummy") "depth-$1$ types".
The ""depth-$2$ $\varphi$-type"" of an interval $I$
with ""left@left context"" and ""right contexts"" $\cL,\cR$
is the tuple $\type{2}{\varphi,\cL,\cR}(I)=(\cL,\cR,\cT,\sB,\sE)$, where 
\begin{itemize}
  \item $\cT = \type{1}{}(I)$,
  \item $\sB = \{ \alpha\in\Depth{1}{\varphi} \::\: 
                  \exists J\prefix I ~~
                  \cL \cdot{1} \type{1}{}(J) \vdash \alpha \big\}$,
  \item $\sE = \{ \alpha\in\Depth{1}{\varphi} \::\:
                  \exists J\suffix I ~~ 
                  \type{1}{}(J) \cdot{1} \cR \vdash \alpha \big\}$.
\end{itemize}
\end{definition}

We give some intuition about the components of a "depth-$2$ $\varphi$-type" 
(the reader can also refer to Figure \ref{fig:depth2-type}).
The component $\cT$ is nothing but the "depth-$1$ type" of the 
reference interval $I$, thus determining which formulas of "depth" 
at most $1$ "hold at" $I$. 
The components $\cL$ and $\cR$ represent the "depth-$1$ types" of 
some intervals adjacent to $I$, to the left and to the right respectively,
and will be used as "contexes" for an operation of
"composition@depth-$2$ composition".
The set $\sB$ represents which subformulas of $\varphi$ of "depth" 
at most $1$ "hold at" some intervals $I'$ that overlap $I$ to the left 
(i.e., such that $\min(I')\le\min(I)\le\max(I')<\max(I)$), provided that the 
"depth-$1$ type" of $K=I'\setminus I$ coincides with the "left context" $\cL$.
The set $\sE$ provides similar information for the intervals $I'$ 
that overlap $I$ to the right and such that $\type{1}{}(I'\setminus I) = \cR$.
As a special case, we observe that when $\cL=\cR=\dummy$, 
one could let $I'$ range over "prefixes" or "suffixes" of $I$, thus determining
which subformulas "hold at" "prefixes" and "suffixes" of the reference interval $I$.
In particular, this can be used to determine the evaluation of $\varphi$
at $I$, and generalizes the second attempt of definition of type 
that we discussed earlier.

\begin{figure}[t]
\begin{minipage}[b][][b]{0.5\textwidth}
\begin{tikzpicture}[scale=0.85]
\begin{scope}
\coordinate (A) at (0,0);
\coordinate (B) at (1.5,0);
\coordinate (C) at (3,0);
\coordinate (D) at (3.5,0);
\coordinate (E) at (4,0);
\coordinate (F) at (3,0);
\coordinate (G) at (3.5,0);
\coordinate (H) at (4,0);
\coordinate (I) at (5.5,0);
\coordinate (J) at (7,0);
\draw [dashed] (A) -- (B) node [above=0.5mm, midway] {$\cL$};
\draw [|-|] (B) -- (I) node [rectangle, fill=white, midway] {$\cT$};
\draw [dashed] (I) -- (J) node [below=1mm, midway] {$\cR$};
\draw [nicecyan, below round brace] ([yshift=-7] A) -- ([yshift=-7] C)
      node [below, midway, yshift=-4] {$\alpha_1$};
\draw [nicecyan, below round brace] ([yshift=-18] A) -- ([yshift=-18] D)
      node [below, midway, yshift=-4] {$\alpha_2$};
\draw [nicecyan, below round brace] ([yshift=-29] A) -- ([yshift=-29] E)
      node [below, midway, yshift=-4] {$\alpha_3$};
\draw [nicecyan] ([xshift=-10, yshift=-20] A) node {\Large $\sB$};
\draw [nicegreen, above round brace] ([yshift=7] H) -- ([yshift=7] J)
      node [above, midway, yshift=4.5] {$\alpha_1$};
\draw [nicegreen, above round brace] ([yshift=18] G) -- ([yshift=18] J)
      node [above, midway, yshift=4.5] {$\alpha_2$};
\draw [nicegreen, above round brace] ([yshift=29] F) -- ([yshift=29] J)
      node [above, midway, yshift=4.5] {$\alpha_3$};
\draw [nicegreen] ([xshift=10, yshift=20] J) node {\Large $\sE$};
\end{scope}
\end{tikzpicture}
\captionof{figure}{Components of a "depth-$2$ type".}\label{fig:depth2-type}
\end{minipage}%
\begin{minipage}[b][][b]{0.5\textwidth}
\begin{flushright}
\begin{tikzpicture}[scale=0.85]
\begin{scope}
\coordinate (A) at (0,0);
\coordinate (B) at (1,0);
\coordinate (C) at (2,0);
\coordinate (D) at (2.5,0);
\coordinate (E) at (3,0);
\coordinate (F) at (4,0);
\coordinate (G) at (8,0);
\end{scope}
\begin{scope}[yshift=-60]
\coordinate (A') at (0,0);
\coordinate (B') at (4,0);
\coordinate (C') at (5,0);
\coordinate (D') at (5.5,0);
\coordinate (E') at (6,0);
\coordinate (F') at (7,0);
\coordinate (G') at (8,0);
\end{scope}
\begin{scope}[yshift=-120]
\coordinate (A'') at (0,0);
\coordinate (B'') at (1,0);
\coordinate (C'') at (2,0);
\coordinate (D'') at (2.5,0);
\coordinate (E'') at (3,0);
\coordinate (F'') at (4,0);
\coordinate (G'') at (5,0);
\coordinate (H'') at (5.5,0);
\coordinate (I'') at (6,0);
\coordinate (J'') at (7,0);
\coordinate (K'') at (8,0);
\end{scope}
\draw [dotted, gray, short] (B) -- (B'');
\draw [dotted, gray, short] (F) -- ([yshift=-35] F'');
\draw [dotted, gray, short] (F') -- (J'');
\draw [white, line width=8pt] (A) -- (G);
\draw [white, line width=8pt] (A') -- (G');
\draw [white, line width=8pt] (A'') -- (K'');
\draw [dashed] (A) -- (B) node [above=0.5mm, midway] {$\cL$};
\draw [|-|] (B) -- (F) node [above=0.5mm, midway] {$\cT$};
\draw [dashed] (F) -- (G) node [above=0.5mm, midway] {$\cR=\cT'\cdot{1}\cR'$};
\begin{scope}
\draw [white, line width=8pt, below round brace] 
      ([yshift=-7] A) -- ([yshift=-7] C);
\draw [white, line width=8pt, below round brace] 
      ([yshift=-14] A) -- ([yshift=-14] D);
\draw [white, line width=8pt, below round brace] 
      ([yshift=-21] A) -- ([yshift=-21] E);
\end{scope}
\draw [nicecyan, below round brace] 
      ([yshift=-7] A) -- ([yshift=-7] C);
\draw [nicecyan, below round brace] 
      ([yshift=-14] A) -- ([yshift=-14] D);
\draw [nicecyan, below round brace] 
      ([yshift=-21] A) -- ([yshift=-21] E);
\draw [nicecyan] ([xshift=-13, yshift=-15] A) node {\Large $\sB$};
\draw [dashed] (A') -- (B') node [above=0.5mm, midway] {$\cL'=\cL\cdot{1}\cT$};
\draw [|-|] (B') -- (F') node [above=0.5mm, midway] {$\cT'$};
\draw [dashed] (F') -- (G') node [above=0.5mm, midway] {$\cR'$};
\begin{scope}
\draw [white, line width=8pt, below round brace] 
      ([yshift=-7] A') -- ([yshift=-7] C');
\draw [white, line width=8pt, below round brace] 
      ([yshift=-14] A') -- ([yshift=-14] D');
\draw [white, line width=8pt, below round brace] 
      ([yshift=-21] A') -- ([yshift=-21] E');
\end{scope}
\draw [nicecyan, below round brace] 
      ([yshift=-7] A') -- ([yshift=-7] C');
\draw [nicecyan, below round brace] 
      ([yshift=-14] A') -- ([yshift=-14] D');
\draw [nicecyan, below round brace] 
      ([yshift=-21] A') -- ([yshift=-21] E');
\draw [nicecyan] ([xshift=-13, yshift=-15] A') node {\Large $\sB\,'$};
\draw [dashed] (A'') -- (B'') node [above=0.5mm, midway] {$\cL$};
\draw [|-|] (B'') -- (J'') 
      node [above=0.5mm, midway, fill=white, rectangle] {$\cT\cdot{1}\cT'$};
\draw [dashed] (J'') -- (K'') node [above=0.5mm, midway] {$\cR'$};
\draw [nicecyan, below round brace] 
      ([yshift=-7] A'') -- ([yshift=-7] C'');
\draw [nicecyan, below round brace] 
      ([yshift=-14] A'') -- ([yshift=-14] D'');
\draw [nicecyan, below round brace] 
      ([yshift=-21] A'') -- ([yshift=-21] E'');
\draw [nicecyan] ([xshift=-13, yshift=-15] A'') node {\Large $\sB$};
\draw [nicecyan, below round brace] 
      ([yshift=-35] A'') -- ([yshift=-35] F'');
\draw [nicecyan] ([xshift=-13, yshift=-37] A'') 
      node {\Large $\leftward{\cup~}\sB_\star$};
\draw [nicecyan, below round brace] 
      ([yshift=-49] A'') -- ([yshift=-49] G'');
\draw [nicecyan, below round brace] 
      ([yshift=-56] A'') -- ([yshift=-56] H'');
\draw [nicecyan, below round brace] 
      ([yshift=-63] A'') -- ([yshift=-63] I'');
\draw [nicecyan] ([xshift=-13, yshift=-60] A'') 
      node {\Large $\leftward{\cup~}\sB\,'$};
\end{tikzpicture}
\end{flushright}
\centering
\captionof{figure}{Composition of "depth-$2$ types".}\label{fig:depth2-composition}
\end{minipage}
\end{figure}

Below, we prove the analogous of Lemmas \ref{lem:depth1-composition} and 
\ref{lem:depth1-evaluation} for "depth"-$2$ types.

\begin{restatable}{lemma}{DepthTwoComposition}\label{lem:depth2-composition}
There is a ""composition@depth-$2$ composition"" operator $\cdot{2}$ 
on "depth-$2$ $\varphi$-types" that is computable in polynomial time 
and such that, for all "contexts" $\cL,\cL',\cR,\cR'$ and 
for all pairs of adjacent intervals $I,I'$, 
if $\cL \cdot{1} \type{1}{}(I) = \cL'$ and 
   $\type{1}{}(I') \cdot{1} \cR' = \cR$, then
\[
  \type{2}{\varphi,\cL,\cR}(I) \cdot{2} \type{2}{\varphi,\cL',\cR'}(I')
  ~=~ 
  \type{2}{\varphi,\cL,\cR'}(I\cup I') \ .
\]
\end{restatable}

\begin{proof}
For the sake of brevity, let
$\sT = \type{2}{\varphi,\cL,\cR}(I)$ and
$\sT' = \type{2}{\varphi,\cL',\cR'}(I')$, 
where $\sT = (\cL,\cR,\cT,\sB,\sE)$, 
$\sT' = (\cL',\cR',\cT',\sB',\sE')$,
$\cL \cdot{1} \cT = \cL'$, and
$\cT' \cdot{1} \cR' = \cR$.
We define the "composition@depth-$2$ composition" as 
\[
  \sT \cdot{2} \sT'
  ~=~ 
  (\cL, \: \cR',\: \cT \cdot{1} \cT',\: \sB\cup\sB'\cup\sB_\star,\: \sE\cup\sE'\cup\sE_\star)
\]
where 
\begin{align*}
\cB_\star &~=~ \{\alpha\in\Depth{1}{\varphi} \::\: \cL'\vdash\alpha\} \\
\cE_\star &~=~ \{\alpha\in\Depth{1}{\varphi} \::\: \cR\vdash\alpha\}
\end{align*}
(see Figure \ref{fig:depth2-composition}).

Note that, thanks to Lemma \ref{lem:depth1-evaluation}, 
the "composition@depth-$2$ composition" $\sT \cdot{2} \sT'$
can be computed in polynomial time given the "types@depth-$2$ types"
$\sT$ and $\sT'$.

Below, we prove that the defined "composition@depth-$2$ composition" $\sT \cdot{2} \sT'$ 
is correct, namely, it coincides with $\type{2}{\varphi,\cL,\cR'}(I\cup I')$.
The latter "type@depth-$2$ type" is of the form $(\cL,\cR',\cT'',\sB'',\sE'')$,
so the first two components of $\sT \cdot{2} \sT'$ are clearly correct.
It remains to prove that $\cT'' = \cT \cdot{1} \cT'$, $\sB'' = \sB\cup\sB'\cup\sB_\star$,
and $\sE'' = \sE\cup\sE'\cup\sE_\star$.
By Lemma \ref{lem:depth1-composition} we have 
$\cT'' = \type{1}{}(I\cup I') 
 = \type{1}{}(I) \cdot{1} \type{1}{}(I') 
 = \cT \cdot{1} \cT'$.
Moreover, by Definition \ref{def:depth2-type}, 
$\sB''$ contains the formulas $\alpha\in\Depth{1}{\varphi}$ 
that satisfy one of the following conditions:
\begin{enumerate}
\item $I''\models\alpha$, for some interval $I''$ that overlaps $I$ to the left 
      (i.e., $\min(I'')\le\min(I)\le\max(I'')<\max(I)$) and such that 
      $\type{1}{}(I''\setminus I)=\cL$.
      
Letting $K=I''\setminus I$ and $J=I''\cap I$ and using Lemma \ref{lem:depth1-evaluation}, 
this condition is equivalent to 
\[
  \type{1}{}(I'') 
  ~=~ \type{1}{}(K) \cdot{1} \type{1}{}(J)
  ~=~ \cL \cdot{1} \type{1}{}(J) ~\vdash~ \alpha
\]
and hence to $\alpha\in\sB$.
\item $I''\models\alpha$, for some interval $I''$ that has $I$ as a "suffix" 
      (i.e., $\min(I'')\le\min(I)\le\max(I'')=\max(I)$) and such that
      $\type{1}{}(I''\setminus I)=\cL$.
      
Letting $K=I''\setminus I$ and using Lemma \ref{lem:depth1-evaluation},
together with the assumptions about the "contexts" $\cL$ and $\cL'$, 
this condition turns out to be equivalent to
\[
  \type{1}{}(I'') 
  ~=~ \type{1}{}(K) \cdot{1} \cT 
  ~=~ \cL \cdot{1} \cT 
  ~=~ \cL' ~\vdash~ \alpha
\]
and hence to $\alpha\in\sB_\star$.
\item $I''\models\alpha$, for some interval $I''$ that contains $I$, 
      overlaps $I'$ to the left
      (i.e., $\min(I'')\le\min(I)\le\max(I)<\max(I'')<\max(I')$), 
      and such that $\type{1}{}(I''\setminus I')=\cL$.

Letting $K=I''\setminus I'$ and $J=I'' \cap I'$, and using again
Lemma \ref{lem:depth1-evaluation} and the assumptions about the "contexts" 
$\cL$ and $\cL'$, this condition turns out to be equivalent to 
\begin{align*}
  \type{1}{}(I'') 
  &~=~ \type{1}{}(K) \cdot{1} \cT \cdot{1} \type{1}{}(J) \\
  &~=~ \cL \cdot{1} \cT \cdot{1} \type{1}{}(J) 
   ~=~ \cL' \cdot{1} \type{1}{}(J) ~\vdash~ \alpha
\end{align*}
and hence to $\alpha\in\sB'$.
\end{enumerate}
We have just shown that $\sB''=\sB\cup\sB'\cup\sB_\star$.
One proves $\sE''=\sE\cup\sE'\cup\sE_\star$ using symmetric arguments.
\end{proof}

\begin{restatable}{lemma}{DepthTwoEvaluation}\label{lem:depth2-evaluation}
For all intervals $I$ and $J$ such that
$\type{2}{\varphi,\dummy,\dummy}(I) = \type{2}{\varphi,\dummy,\dummy}(J)$
and for every $\BEpi$ formula $\varphi$ of "depth" at most $2$,
we have
$\cS,I\models\varphi$ iff $\cS,J\models\varphi$.
Moreover, whether $\cS,I\models\varphi$ holds can be decided in
polynomial time from the given "type@depth-$2$ type"
$\type{2}{\varphi,\dummy,\dummy}(I)$.
\end{restatable}

\begin{proof}
Let $\varphi$ be a formula of "depth" at most $2$
and let $I$ be an interval with "depth-$2$ type"
$(\dummy,\dummy,\cT,\sB,\sE)$, where both "left@left context" 
and "right contexts" are $\dummy$.

If $\varphi$ has "depth" smaller than $2$, 
then by Lemma \ref{lem:depth1-evaluation}
the component $\cT=\type{1}{}(I)$  
already determines (effectively in polynomial time)
whether $\cS,I\models\varphi$.

Otherwise, if $\varphi$ has "depth" $2$ 
and is of the form $\hsB\alpha$, then 
$\cS,I\models\hsB\alpha$ iff there is a "proper prefix" $J$ 
of $I$ such that $\cS,J\models\alpha$.
Since $\alpha\in\Depth{1}{\varphi}$, the latter condition 
is equivalent to $\dummy \cdot{1} \type{1}{}(J) \vdash \alpha$, 
and hence $\cS,I\models\hsB\alpha$ iff $\alpha\in\sB$.
The case of $\varphi=\hsE\alpha$ is similar, but uses 
the component $\sE$.

Finally, Boolean combinations of the previous formulas are 
"evaluated" homomorphically. 
\end{proof}

\knowledge{notion}
    | \encodedword*{\cS}
	| encoding
	| encodings
	| encodes
	| encode
	| encoded
\knowledge{notion}
    | \sT \vdashbis* \psi
    | \vdashbis
\knowledge{notion}
	| dummy depth-$2$ type
	| dummy depth-$2$ types
\knowledge{notion}
	| \hsG
	| \hsG*

\subsection{Satisfiability procedure}\label{subsec:satisfiability}

As a warm-up, let us first describe the "satisfiability" procedure for a formula
of "depth" at most $2$; later we will generalize this to a formula in "shallow normal form".

Let us fix a $\BEpi$ formula $\psi$ of "depth" at most $2$.
Deciding "satisfiability" of $\psi$ can be done in polynomial space, 
by reducing to non-emptiness of a language recognized by a suitable finite state automaton. 
\AP
To formalize the construction of the automaton from the given formula $\psi$, 
it is convenient to encode an "interval structure" $\cS=\intervalstructure$ 
over the "signature" $\signature$ by the finite word 
$""\encodedword*{\cS}"" = a_0 \dots a_{\max(N)}$ over the alphabet 
$\wp(\signature)$, where $a_i = \sigma(i)$ for all $i\in \timedomain$ 
(recall that $\timedomain$ is a finite prefix of the natural numbers). 

\begin{restatable}{lemma}{Automaton}
\label{lem:automaton}
Given a $\BEpi$ formula $\psi$ of "depth" 
at most $2$, one can compute in polynomial space%
\footnote{By computing an automaton in polynomial space we mean that 
    its initial states, final states, and transitions can be enumerated
    in polynomial space.
    The enumeration procedures can be used within other algorithms of 
    similar complexity, e.g., to test emptiness of the recognized language.}
a finite state automaton $\cA_\psi$ that accepts 
all and only the "encodings" $\encodedword{\cS}$ of the "interval structures" $\cS$ 
such that $\cS,I\models\psi$, where $I$ is the largest interval of $\cS$.
\end{restatable}

\begin{proof}[Proof sketch]
The construction of $\cA_\psi$ is quite standard, 
as it is the cascade product of three automata:
\begin{enumerate}
\item a deterministic automaton that 
      computes in its states the "depth-$1$ types" of intervals 
      corresponding to "prefixes" of the input,
\item a co-deterministic automaton that 
      computes in its states the "depth-$1$ types" of intervals 
      corresponding to "suffixes" of the input,
\item a deterministic automaton that 
      computes the "depth-$2$ $\psi$-type" of "prefixes" 
      of the input, with a constant "dummy" "left context" 
      and "right contexts" given by the states of the 
      previous automaton.
\end{enumerate}
Transitions of these automata are defined using compositional
properties of "depth-$1$@depth-$1$ types" and "depth-$2$ types" 
(Lemmas \ref{lem:depth1-composition} and \ref{lem:depth2-composition}). 

Below, we provide full details for the construction of $\cA_\psi$.
\AP
Like we have done for "depth-$1$ types", we introduce 
""dummy depth-$2$ types"" for abstracting an empty interval:
these are tuples of the form $(\cL,\cR,\dummy,\emptyset,\emptyset)$, 
where $\cL$ and $\cR$ are "left@left context" and "right contexts"
and $\dummy$ is the "dummy depth-$1$ type"
(of course, there is exactly one "dummy depth-$2$ type" for each
choice of the "left@left context" and "right contexts"). As usual, 
a "dummy type" behaves as an identity w.r.t.~composition
with a "depth-$2$ type", provided the "contexts" are compatible.
We shall also use a generalization of the relation $\vdash$ that 
works with "depth-$2$ types". 
\AP
Precisely, given a "depth-$2$ type" $\sT$,
we write $""\sT \vdashbis* \psi""$ 
whenever $\cS,I\models\psi$ for some (or, equally, for every) 
interval $I$ such that $\type{2}{\varphi,\dummy,\dummy}(I)=\sT$.
\begin{itemize}
\item the alphabet $A$ consists of subsets of the "signature" $\signature$;
\item the state space $Q$ consists of triples $q=(\cL,\cR,\sT)$, where 
      $\cL,\cR$ are a "depth-$1$ types" and 
      $\sT$ is a "depth-$2$ $\psi$-type"
      with $\dummy$ as "left context" and $\cR$ as "right context";
\item the set $I$ of initial states consists of triples $q=(\cL,\cR,\sT)$, 
      where $\cL=\dummy$ is the "dummy depth-$1$ type" and 
      $\sT=(\cL,\cR,\dummy,\emptyset,\emptyset)$ 
      is a "dummy depth-$2$ type";
\item the set $F$ of final sates consists of triples 
      $q=(\cL,\cR,\sT)$,
      with $\cR=\dummy$ and $\sT \vdashbis \psi$;
\item the set $T$ of transition rules consists of the 
      triples $(q,a,q')$, with       
      $q=(\cL,\cR,\sT)$, $a\subseteq\signature$, and $q'=(\cL',\cR',\sT')$,
      such that
      $\cL' = \cL \cdot{1} \type{1}{}(I_a)$,
      $\cR = \type{1}{}(I_a) \cdot{1} \cR$, and
      $\cT' = \sT \cdot{2} \type{2}{\psi,\cL,\cR'}(I_a)$,   
      where $I_a$ denotes the singleton interval labelled by 
      the set $a$ of propositional letters.
\end{itemize}
It is worth noting that the automaton $\cA_\psi$ is unambiguous, 
namely, it admits at most one successful run on each input. 

We now claim that, on every input $\encodedword{\cS} = a_0\dots a_{n-1}$, 
the only possible runs of $\cA_\psi$ that start and end in arbitrary states
(not necessarily initial or final ones) are of the form 
\[
  q_0 \trans{a_0} q_1 \trans{a_1} \dots\dots \trans{a_{n-1}} q_n
\]
with $q_i=(\cL_i,\cR_i,\sT_i)$ 
such that, for all $i=0,\dots,n$,
\begin{enumerate}
\item $\cL_i = \cL_0 \cdot{1} \type{1}{}([0,i-1])$,
\item $\cR_i = \type{1}{}([i,n-1]) \cdot{1} \cR_n$,
\item $\sT_i = \sT_0 \cdot{2} \type{2}{\psi,\dummy,\cR_i}([0,i-1])$.
\end{enumerate}
Each of the above properties can be verified using a simple induction,
either from smaller to larger $i$'s or vice versa (we omit the tedious details).

From the properties stated in items 1., 2., 3.~and 
the definitions of initial and final states, it immediately follows that 
$\cA_\psi$ admits a successful run on $\encodedword{\cS}$ 
if and only if $\cS,[0,n-1]\models\psi$.

Finally, as for the complexity of constructing $\cA_\psi$, we
recall from Lemmas \ref{lem:depth1-composition} and \ref{lem:depth2-composition}
that "depth-$1$@depth-$1$ types" and "depth-$2$ types" 
can be enumerated in polynomial space, and can be 
"composed@depth-$1$ composition" in polynomial time. 
This implies that the initial states and the transitions of $\cA_\psi$ 
can be enumerated in polynomial space.
To enumerate the final states, it suffices to test properties like 
$\sT \vdashbis \psi$, for a given "depth-$2$ type" $\sT$. This can 
be done in polynomial time thanks to Lemma \ref{lem:depth2-evaluation}.
\end{proof}

The fact that the automaton $\cA_\psi$ above 
can be constructed from $\psi$ in polynomial space, implies that (non-)emptiness of 
the recognized language can also be decided in polynomial space w.r.t.~$|\psi|$. 
In its turn, this shows that the satisfiability of a $\BEpi$ formula 
$\psi$ of "depth" at most $2$ can be decided in polynomial space.

\medskip
To conclude the proof of Theorem \ref{thm:complexity} it remains to 
reduce the satisfiability problem for a $\BEpi$ formula $\psi$ in 
"shallow normal form" 
to the non-emptiness problem of an automaton $\cA_\psi$ 
that is computable from 
$\psi$ in exponential space.
For this, it suffices to recall that $\psi$ must be of the form 
$\varphi \:\wedge\: \hsGu\xi$, where both $\varphi$ and $\xi$ 
are $\BEpi$ formulas of "depth" at most $2$.
One uses Lemma \ref{lem:automaton} to construct the automata
$\cA_\varphi$ and $\cA_{\neg\xi}$, whose languages contain 
"encodings" of "models" of $\varphi$ and $\neg\xi$, respectively.
From $\cA_{\neg\xi}$, one can efficiently construct an automaton
$\cA_{\hsG\neg\xi}$ recognizing the language of words with 
infixes accepted by $\cA_{\neg\xi}$, thus "encoding" "models" of 
\AP
$\mathop{""\hsG*""}\neg\xi$ ($=\neg\hsGu\xi$).
One then complements the latter automaton to obtain an automaton 
$\cA_{\hsGu\xi}$ accepting the "encodings" of "models" of $\hsGu\xi$. 
Note that the latter step can be performed in exponential space in 
the size of $\xi$, by using an online version of the classical 
subset construction.
Finally, one computes the product of the automata $\cA_\varphi$ and 
$\cA_{\hsGu\xi}$, so as to recognize the language of "encodings" of 
"models" of $\psi = \varphi \:\wedge\: \hsGu\xi$.
It follows that non-emptiness of the latter language can be decided 
in exponential space w.r.t.~the size of the original formula $\psi$.
\qed


\section{Conclusions}\label{sec:conclusions}
We have settled the question of whether the logic $\BE$, 
interpreted over "homogeneous" "interval structures", 
admits an elementary "satisfiability" problem.
We have actually answered the question by giving an 
optimal $\expspace$ decision procedure
($\expspace$-hardness was shown in \cite{DBLP:journals/tcs/BozzelliMMPS19}).
As a by-product result, we have also devised a 
"normal form@shallow normal form" 
for $\BE$ formulas that enforces a small bound to 
the number of nested modalities, 
while "preserving satisfiability@equi-satisfiable". 
Quite suprisingly such a "normal form@shallow normal form" 
can be computed in polynomial time from arbitrary 
$\BE$ formulas, using a series of rewriting steps 
reminiscent of a quantifier elimination technique 
a-la Scott. 

As for future work, one could try to see whether 
similar techniques are applicable to extensions 
of $\BE$ with modalities based on other Allen's
interval relations (e.g., overlap, meet, the 
inverses of the "prefix" and "suffix" relations, etc.).


\bibliographystyle{alpha}
\bibliography{bib2}

\end{document}